%% file: main.tex
\documentclass[10pt]{article}

\usepackage[margin=1in]{geometry}

\usepackage[T1]{fontenc}
\usepackage[scaled=0.95]{inconsolata}   %
\usepackage[final]{microtype}

\input{packages}

\input{macros}

\newtheorem{proposition}{Proposition}
\newtheorem{observation}{Observation}

\newtheorem{claim}{Claim}

\newtheorem{definition}{Definition}

\newtheorem{example}{Example}

\theoremstyle{remark}

\newtheorem{remark}{Remark}

\crefname{definition}{Definition}{Definitions}
\Crefname{definition}{Definition}{Definitions}

\usepackage{tikz}
\usetikzlibrary{positioning, shapes, arrows.meta, calc, patterns.meta, patterns, decorations.pathreplacing, backgrounds}

\definecolor{frameblue}{RGB}{0, 155, 220}
\definecolor{innerorange}{RGB}{255, 120, 30}

\usepackage{todonotes}

\usepackage{color-edits}
\addauthor[TODO]{todo}{orange}
\addauthor[DC]{dc}{blue}
\addauthor[NS]{ns}{teal}
\addauthor[LC]{lc}{olive}

\input{custom}

\title{On Incentivized Exploration beyond Bayesianism and Full-Information}

\usepackage{authblk}

\author[1]{Dimitar Chakarov}
\author[2]{Lee Cohen}
\author[1]{Nathan Srebro}
\affil[1]{Toyota Technological Institute at Chicago}
\affil[2]{Stanford University}

\date{\today}

\begin{document}

\maketitle

\begin{abstract}
    \input{abstract}
\end{abstract}

\input{sec-introduction-lc}
\input{sec-bic}
\input{sec-policies}
\input{sec-main-text}

\bibliographystyle{plainnat}
\bibliography{bibliography}

\appendix
\newpage
\crefalias{section}{appendix}
\crefalias{subsection}{appendix}
\crefalias{subsubsection}{appendix}

\input{app-notation}
\input{app-kremer-policy}
\input{app-sic-example}
\input{app-deterministic-behavioral-policies}
\input{app-proofs}
\input{app-eps-ic}
\input{app-per-step-pareto-optimality}

\end{document}

%% file: packages.tex
\usepackage{graphicx}
\usepackage{subcaption}
\usepackage{booktabs}    %
\usepackage{tabulary}
\usepackage{colortbl}
\usepackage{float}       %

\usepackage{amsmath}
\usepackage{amsthm}
\usepackage{amsfonts}
\usepackage{amssymb}
\usepackage{mathrsfs}
\usepackage{mathtools}
\usepackage{thmtools}
\usepackage{bm}          %

\usepackage[dvipsnames]{xcolor}
\colorlet{toastedorange}{orange!50!BrickRed}
\colorlet{richblue}{blue!75!black}
\colorlet{purpur}{violet!50!magenta}

\usepackage{algorithm}
\usepackage{algpseudocode}
\renewcommand{\algorithmicrequire}{\textbf{Input:}}
\renewcommand{\algorithmicensure}{\textbf{Output:}}

\usepackage[square,sort&compress]{natbib}
\usepackage{xurl}
\usepackage[hidelinks,backref=page,colorlinks=true]{hyperref}
\hypersetup{
    colorlinks=true,
    linkcolor=toastedorange,
    citecolor=richblue,
    urlcolor=purpur
}
\usepackage[capitalise,nameinlink]{cleveref}    %
\let\cref\Cref

\usepackage{setspace}
\usepackage[shortlabels]{enumitem}
\setlist[enumerate]{noitemsep,topsep=0pt,parsep=0pt,partopsep=0pt}
\setlist[itemize]{noitemsep,topsep=0pt,parsep=0pt,partopsep=0pt}

\usepackage{fancyhdr}

\usepackage{xspace}           %
\usepackage{nicefrac}         %
\usepackage{fontawesome}
\usepackage[normalem]{ulem}

%% file: macros.tex
\DeclareMathOperator{\supp}{supp}
\DeclareMathOperator*{\argmax}{arg\,max}

\newcommand{\calA}{\mathcal{A}}

\newcommand{\calE}{\mathcal{E}}
\newcommand{\calF}{\mathcal{F}}

\newcommand{\calH}{\mathcal{H}}

\newcommand{\calM}{\mathcal{M}}

\newcommand{\calU}{\mathcal{U}}

\newcommand*{\bfa}{\mathbf{a}}

\DeclarePairedDelimiter{\bracket}{[}{]}
\DeclareMathOperator*{\Expectation}{\mathbb{E}}
\newcommand{\EE}[2][]{\Expectation_{#1}\bracket*{#2}}

\renewcommand{\Pr}[2][]{\mathbb{P}_{#1}\left[#2\right]}

\newcommand{\bigo}[1]{\mathrm{O}\!\left(#1\right)}

\newcommand{\littleo}[1]{\mathit{o}\!\left(#1\right)}
\newcommand{\bigomega}[1]{\Omega\!\left(#1\right)}

\newcommand{\bigtheta}[1]{\Theta\!\left(#1\right)}

\newcommand{\abs}[1]{\left|#1\right|}

\newcommand{\ds}{\displaystyle}           %

%% file: custom.tex
\newcommand{\paragraphname}[1]{\vspace{0.75em}\noindent\textbf{#1.}\ }

\newcommand{\BIC}{\ensuremath{\operatorname{BIC}}\xspace}
\newcommand{\IC}{\ensuremath{\operatorname{IC}}\xspace}
\newcommand{\SIC}{\ensuremath{\operatorname{SIC}}\xspace}
\newcommand{\oSIC}{\ensuremath{\operatorname{(S)IC}}\xspace}    %
\newcommand{\PO}{\ensuremath{\operatorname{PO}}\xspace}

\newcommand{\epsIC}{\ensuremath{\operatorname{\varepsilon-IC}}\xspace}
\newcommand{\epsPO}{\ensuremath{\operatorname{\widetilde{\varepsilon-PO}}}\xspace}
\newcommand{\hepsPO}{\ensuremath{\operatorname{\varepsilon-PO}}\xspace}

\newcommand{\pip}{\pi_{\mathbf{p}}}

\newcommand{\pia}{\pi_{\mathbf{a}}}

\newcommand{\external}{\mathbf{ext}}

\newcommand{\reg}{\operatorname{Reg}}

\newcommand{\Bern}[1]{\operatorname{Bern}\left(#1\right)}
\newcommand{\Unif}[1]{\calU\left[#1\right]}

\crefname{protocol}{Protocol}{Protocols}
\Crefname{protocol}{Protocol}{Protocols}

\newfloat{protocol}{tbp}{lop}
\floatname{protocol}{Protocol}
\crefname{protocol}{Protocol}{Protocols}
\Crefname{protocol}{Protocol}{Protocols}

%% file: abstract.tex
We extend Incentive Compatible Exploration beyond the Bayesian full-information setting of~\citet{kremer2014implementing}. We consider agents that may possess external information unknown to the principal. We show such settings require new notions of incentivized exploration, as well as going beyond a Bayesian perspective, and we introduce a definition where agents choose any reasonable (undominated) action. Furthermore, our framework provides for a more robust treatment of ties, and  extends to settings where agents lack a single common prior and instead only know that reward distributions belong to a collection of potential priors.

%% file: sec-introduction-lc.tex
\section{Introduction}
\label{section:introduction}

Many learning systems improve by recommending actions to successive users (agents), but the system does not execute those actions itself.
Instead, each round is carried out by a fresh myopic agent who chooses what benefits them immediately, such as in route guidance and recommendations.
This setting is naturally modeled as a \textit{multi-armed bandit} problem (each arm is an action) with autonomous agents, and was introduced as Bayesian incentive-compatible (\BIC) exploration in~\citet{kremer2014implementing}. We refer the reader to~\citet{Slivkins-book} and \citet{Lattimore-Szepesvari-Book} for background on multi-armed bandit.
This seemingly small change is consequential: agents are autonomous and seek to maximize their own reward (utility), without caring about how exploration benefits future agents.
In this work, we further extend the model by allowing agents to condition on \textbf{external information} that the principal (social planner) does not necessarily observe nor control.

We focus on settings without monetary transfers, where the principal can influence agent choices only through messages that leverage \textit{information asymmetry}, i.e., things the principal knows and the agent does not.\footnote{Another approach to bridge this gap relies on monetary transfers, where the principal explicitly compensates agents for the exploration incurred by choosing potentially sub-optimal arms~\citep{frazier2014incentivizing}. }
As a result, the principal faces a learning problem where exploration must be induced through communication, despite users’ incentives and external information.
Formally, in each round, the principal sends a private message to a fresh agent who chooses an action and receives its realized reward, and the principal observes both. The principal's goal is to maximize cumulative reward (which is equivalent to maximizing social welfare or minimizing regret).

We assume deterministic outcomes~\citep{kremer2014implementing,cohen2019optimal,bahar2020fiduciary}: each action has an unknown fixed reward drawn once from the prior the first time it is played, and the same reward is realized on every subsequent play.\footnote{While some prior works study stochastic rewards, we focus on deterministic outcomes because the incentive and information issues that drive our results already arise in this simplest setting. Moreover, all our notions extend to stochastic outcomes in the standard way. }

Classical works on \BIC exploration~\citep{kremer2014implementing,mansour2020bayesian,mansour2022bayesian,cohen2019optimal,sellke2021price} show that recommendations alone can induce (asymptotically) optimal exploration under a tightly structured model. First, agents do not observe the full action-reward history and have no external information, so their decision-relevant knowledge is controlled by what the principal chooses to reveal, and in particular, the principal can anticipate what the agent knows. Second, the uncertainty is a Bayesian instance given by a single common prior shared by the principal and agents,  so the principal can predict how agents evaluate actions. Third, agents are assumed to follow recommendations even when they are merely incentive compatible (\IC) in the weak sense given their knowledge, that is, they comply in ties. Under these assumptions, and because the principal effectively has \emph{one-sided information asymmetry} (where the agents' knowledge is a subset of the principal's knowledge), the principal can recommend actions that are \IC (and therefore followed) while still driving exploration. This corresponds to case (a) in \cref{fig:information-asymmetry}.

Now consider external information. Effective \BIC exploration might still be possible when the principal has access to the external information of the agent, e.g., the agents observe the history with probability $1/2$ and the principal knows which agent observes the history (case (b) in~\cref{fig:information-asymmetry}). However, external information unknown to the principal undermines the foundation of \BIC exploration. For example, a navigation app learns a route’s delay only when users take it, but drivers may also hear a traffic report on the radio and ignore an exploratory suggestion; diners hear restaurant reviews from friends who visited the restaurant before and decide to go to a different restaurant than the one recommended.
When an agent conditions on external information which is unknown to the principal, the principal can no longer certify that a recommendation is \IC for that agent.
This changes the problem from one-sided information asymmetry into \emph{two-sided information asymmetry}, where each side might know something the other does not. See \cref{fig:information-asymmetry} for the spectrum of information settings we study.
Consequently, \IC defined relative to the principal’s information becomes fragile, and the principal cannot guarantee that agents will follow IC recommendations. This is not merely a technical nuisance but a modeling mismatch: in realistic settings, the principal’s recommendation competes with information sources the principal does not observe and cannot control.

\begin{figure}[tbp]
    \centering
    \begin{tikzpicture}[scale=0.8, every node/.style={scale=0.9}]
        \begin{scope}[shift={(0,0)}]
          \draw[toastedorange, line width=1.5pt] (0,0) circle (1.2);
          \draw[richblue, line width=1.5pt] (-0.5,-0.5) circle (0.1);
          \node[below, align=left] at (0,-1.5) {(a) Full Information \\ \phantom{(a) }Setting (\cref{example:bayesian-ext-info-full-info})};
        \end{scope}

        \begin{scope}[shift={(6,0)}]
          \draw[toastedorange, line width=1.5pt] (0,0) circle (1.2);
          \draw[richblue, line width=1.5pt] (-0.5,-0.5) circle (0.4);
          \node[below, align=left] at (0,-1.5) {(b) Observable External \\ \phantom{(b) }Information Setting \\ \phantom{(b) }(\cref{example:bayesian-ext-info-principal-knows})};
        \end{scope}

        \begin{scope}[shift={(12,0)}]
          \draw[toastedorange, line width=1.5pt] (0,0) circle (1.2);
          \draw[richblue, line width=1.5pt] (-1.0,-0.5) circle (0.5);
          \node[below, align=left] at (0.4,-1.5) {(c) Private External \\ \phantom{(c) }Information Setting \\ \phantom{(c) }(\cref{example:bayesian-ext-info})};
        \end{scope}
    \end{tikzpicture}
    \caption{Information asymmetry. The big Orange circles represent the principal's knowledge. The small Blue circles represent an agent's knowledge. Cases (a) and (b) show one-sided information asymmetry, where the principal knows everything the agent knows, even if the agent receives external information (as is in case (b)). Case (c) shows two-sided information asymmetry, where both know something the other does not.}
    \label{fig:information-asymmetry}
\end{figure}
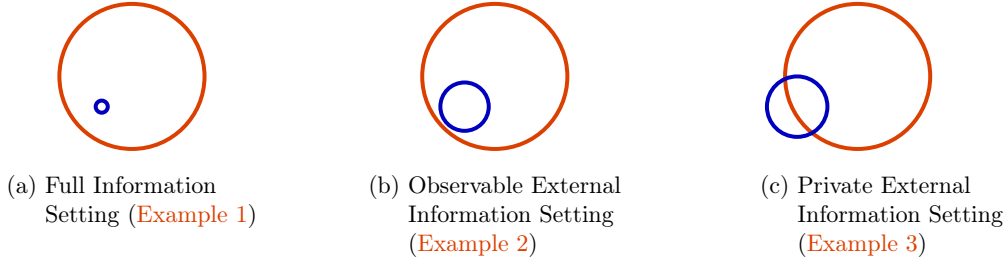

Moreover, external information pushes the problem beyond the usual “Bayesian instance” interpretation in a structural way. Even if the environment is governed by a single common prior that all agree on, later agents need not have a well-defined distribution over earlier agents’ actions, because those actions may depend on external signals that later agents do not observe. As a result, a later agent cannot, in general, assign probabilities to what earlier agents would have done, and therefore cannot form a well-defined Bayesian posterior about the history induced by the principal’s policy. This breaks the standard Bayesian modeling viewpoint in which the agent’s beliefs are pinned down by a shared prior and a known description of how earlier agents respond, and motivates a non-Bayesian approach that does not rely on a shared posterior over past play.

Motivated by this, we do not attempt to predict which action an agent will take from a recommendation. Instead, we adopt a minimal behavioral requirement that remains meaningful under external information: a selfish agent should avoid actions that are clearly inferior given what they know. Concretely, we assume agents choose only actions that are not \textit{dominated} under their information. This “reasonable behavior” assumption is weak enough to accommodate arbitrary tie-breaking and (private) external information signals, yet strong enough to support principled guarantees for the principal’s learning objective. Moreover, the challenge of relaxing standard (yet strong) behavioral assumptions has been identified by~\citet{Slivkins-book} as a key direction for making agent-based bandit models more realistic.

We formalize this idea through \textit{Pareto-optimal (PO) behavior policies}. Roughly, a PO policy allows an agent to pick any action that is not strictly dominated, meaning there is no other action that is at least as good under all plausible scenarios consistent with the agent’s information and strictly better under some. This notion captures the autonomy created by external information: agents may reject a recommendation because of a private signal, may hold a subjective prior that differs from the principal’s, and may break ties arbitrarily. The principal's problem, therefore, changes. Rather than sending recommendations and relying on
agents being obedient and a structured model, the principal’s goal is to send information so that \textit{every} action that remains undominated for the agent still yields sufficient exploration progress.

Our main message is that exploration can still be possible in this external-information regime. We show that there are instances in which the principal can guarantee sublinear regret against PO behavior by using general policies that are designed for the agent’s autonomy, rather than for obedience to recommendations. This approach also resolves two additional modeling issues in a unified way: it is robust to arbitrary tie-breaking, and it naturally extends to non-Bayesian instances in which agents’ beliefs are not described by a single shared prior. In particular, unlike most prior works that fix an incentive notion and design algorithms to optimize regret, we formalize and compare behavior and principal policy classes, and characterize when sublinear regret is possible under each. We remark that the special case of external information where the information was about the past was left as an open question in~\cite{cohen2019optimal}.

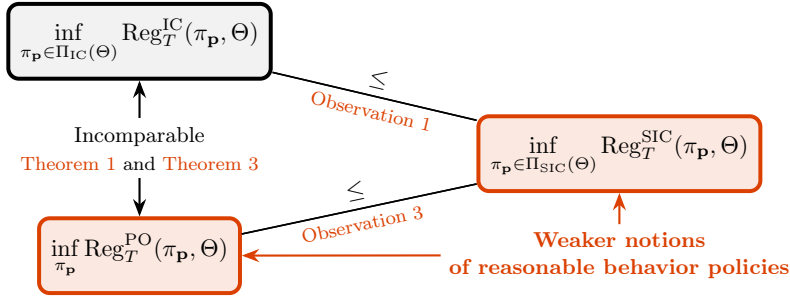
\begin{figure}[tbp]
    \centering
    \begin{tikzpicture}[
        scale=0.9,
        transform shape,
        qty/.style={rectangle, draw=black, line width=1.5pt, rounded corners, minimum size=1.1cm, inner sep=5pt, fill=gray!10},
        compare/.style={semithick, -},
        doublearrow/.style={thick, <->, >=Stealth},
        highlight/.style={qty, fill=toastedorange!12, draw=toastedorange}, %
        note/.style={text=toastedorange, font=\bfseries\small, align=center},
        node distance=0.5cm and 3cm
    ]

        \node[qty] (B) {$\ds\inf_{\pip \in \Pi_{\IC}(\Theta)} \reg^{\IC}_T(\pip, \Theta)$};
        \node[highlight] (C) [below=2cm of B] {$\ds\inf_{\pip} \reg^{\PO}_T(\pip, \Theta)$};
        \node[highlight] (D) [below right=of B] {$\ds\inf_{\pip \in \Pi_{\SIC}(\Theta)} \reg^{\SIC}_T(\pip, \Theta)$};

        \draw[compare] (D) -- node[midway, above, sloped] {$\leq$} (B);
        \draw[compare] (D) -- node[midway, below, sloped] {\footnotesize \cref{claim:sic-is-ic}} (B);
        \draw[compare] (D) -- node[midway, above, sloped] {$\leq$} (C);
        \draw[compare] (D) -- node[midway, below, sloped] {\footnotesize \cref{observation:sic-po}} (C);

        \draw[doublearrow] (B) -- node[midway, fill=white, align=center] {\small Incomparable \\ \footnotesize \cref{claim:ic-po-full-info} and \cref{claim:bayesian-external-information}} (C);

        \node[note] (weakest) [below=0.46cm of D] {\small  Weaker notions \\ of reasonable behavior policies};

        \draw[->, thick, toastedorange, >=Stealth] (weakest) -- (C);
        \draw[->, thick, toastedorange, >=Stealth] (weakest) -- (D);
    \end{tikzpicture}
    \caption{Comparison of regret notions. In the general external information setting, \IC regret and \PO regret are incomparable. In the full information setting, $\inf_{\pip \in \Pi_{\IC}} \reg^{\IC}_T(\pip, \Theta) \leq \inf_{\pip} \reg^{\PO}_T(\pip, \Theta)$, however, there exists a Bayesian external information instance in which no \IC advice policy achieves sublinear regret but there exists a general policy that achieves sublinear \PO regret.
    }
    \label{figure:comparison-of-regret}
\end{figure}

\subsection{Our Contributions}
\begin{itemize}
    \item \textbf{General framework for exploration with external information.} Following a brief exposition of the classic \BIC setting in \cref{section:kremer-review}, we study exploration when agents may have external information and may break ties arbitrarily, and show that the principal can no longer rely on \BIC exploration (\cref{section:external-information}) to achieve sublinear regret. Moreover, this motivates a non-Bayesian approach that does not rely on a shared posterior over past play. We therefore introduce a general framework for exploration with selfish agents that separates the principal's policy from the agents' behavior policy in \cref{section:behavior-policies} that accommodates external information and non-Bayesian instances.
    \item \textbf{Agent behaviors policies and principal policies.}
    We introduce different types of \emph{principal policies} (with arbitrary signaling) and \emph{agent behavior policies} (\cref{section:principal-and-agent-behavior-policies}).
    On the agent side, we formalize three notions of behaviors given their knowledge: Incentive Compatible (\IC), which assumes any reasonable recommendation is followed, Strong Incentive Compatible (\SIC), which assumes that only a unique reasonable recommendation is followed, and \PO behavior, which assumes agents select a reasonable actions and that the principal sends messages rather than recommendations.
    On the principal side, we define the corresponding classes of advice policies (\IC, \SIC), and allow general principal policies that can signal beyond action recommendations. This separation lets us analyze exploration under external information, heterogeneous priors, and tie-breaking.  Namely, we  ask when sublinear regret is achievable under different behavioral notions.
    \item \textbf{External information can rule out principal advice policies, while general still works.}
    In the external-information setting, we show that \SIC advice can be infeasible and \IC advice policies can incur linear regret, yet a general principal policy can still achieve sublinear \PO regret by signaling appropriately (\cref{claim:bayesian-external-information}).
    \item \textbf{Relations between different types of policies.}
    In \cref{section:separations}, we compare behavioral notions in terms of best achievable regret.
    We show that \SIC can be overly restrictive: we provide instances where sublinear regret is achievable under \IC or \PO behavior but impossible under \SIC behavior (\cref{table:summary-sic-ic-po-results}).
    Moreover, \IC and \PO regret are incomparable.
    There are instances where  advice policies achieve sublinear regret, but no general principal policy can guarantee sublinear \PO regret, and there are also instances where a general  principal policy attains sublinear \PO regret while every advice policy has linear \IC regret (\cref{claim:bayesian-external-information,claim:non-bayesian-full-information-no-po}).
    This shows that ``recommendations that an agent should follow'' (\IC%
    ) and ``freedom to choose among undominated actions'' (\PO) lead to genuinely different exploration guarantees.

    \item \textbf{Approximate  Pareto-optimality.} We extend the framework to approximate Pareto-optimality (\hepsPO), modeling agents who select approximately reasonable actions due to friction (\cref{section:approximate-pareto-optimality}). We exhibit an instance that admits a general policy with sublinear \PO regret, however, every general policy has linear \hepsPO regret.
\end{itemize}

\begin{table}[tbp]
    \centering
    \caption{
    Achievability of sublinear regret across information settings and regimes. Here \faTimes\ denotes that there exists an instance for which every policy has linear regret of the corresponding type. \faCheck\ indicates that there exists an instance for which there is a policy with sublinear regret of the corresponding type.
    }
    \label{table:summary-sic-ic-po-results}
    \begin{tabular}{@{}lccccc@{}}
        \toprule
        \textbf{Result} & \textbf{Regime} & \textbf{Information Setting} & \textbf{SIC} & \textbf{IC} & \textbf{PO} \\
        \midrule
        \cref{claim:bayesian-full-information} & Bayesian & Full & \faTimes & \faCheck & \faCheck \\
        \cref{claim:bayesian-external-information} & Bayesian & External & \faTimes & \faTimes & \faCheck \\
        \cref{claim:bayesian-external-information-no-sublinear} & Bayesian & External & \faTimes & \faTimes & \faTimes \\
        \cref{claim:non-bayesian-full-information} & Non-Bayesian & Full & \faTimes & \faCheck & \faCheck \\
        \cref{claim:non-bayesian-full-information-no-po} & Non-Bayesian & Full & \faTimes & \faCheck & \faTimes \\
        \bottomrule
    \end{tabular}
\end{table}

\subsection{Related Work}
\label{section:related-work}

Learning with strategic agents is a rapidly growing research area \citep{hardt2016strategic,zhang2022efficient,ben2023learning}.
Our work is closely related to Bayesian incentive-compatible exploration under information asymmetry. \citet{kremer2014implementing} introduced the classical Bayesian incentive-compatible exploration model and derived the optimal policy for two deterministic actions. \citet{mansour2020bayesian} designed a bandit algorithm with asymptotically optimal regret in the case of stochastic actions, and defined the Bayesian full information version of strong incentive compatibility (\SIC). \citet{mansour2022bayesian} extended the model to multi-agent games, while \citet{mansour2018competing} studied two competing planners. \citet{cohen2019optimal} designed optimal exploration policies for the principal, when rewards are deterministic and have limited support. These prior works assumed that the principal knows all the information available to the agents. However, this assumption is not suited for real-world settings where agents observe external information that is unknown to the principal. To this end, \citet{bahar2019social} explored a restricted form of exploration with external information, where agents can observe the action and reward of the agent immediately before them.

\citet{immorlica2020incentivizing} studied incentivized exploration with disclosure policies, where the platform deliberately reveals selected subhistories of past outcomes. Contrast this with our model, where the principal can send any arbitrary message and agents may receive other information about the history from their external information. They consider stochastic Bernoulli rewards as opposed to our deterministic rewards setup. Similarly to our work, they formalize a specific type of agent behavior based on the observed empirical outcomes for each arm.

In a different line of work, \citet{frazier2014incentivizing} considered incentivizing exploration through monetary transfers, while \citet{che2015optimal} consider a setting with two binary-valued actions and a continuum of agents.
\citet{bahar2016economic} investigated a version of external information where agents can observe their neighbors in a social network. More generally, incentivized exploration has been studied broadly in multi-armed bandits \citep{bahar2020fiduciary,slivkins2017incentivizing,immorlica2019bayesian,immorlica2020incentivizing} and in Markov decision processes in the context of reinforcement learning \citep{simchowitz2024exploration}.
The multi-armed bandits (MAB) problem has been well studied in operations research, machine learning, and computer science to model the exploration-exploitation trade-off~\citep{Slivkins-book,Lattimore-Szepesvari-Book}.
In our setting, if the principal could choose actions without having to interact with myopic agents, the learning problem reduces to MAB.

The interaction between the principal and an agent corresponds to a version of \emph{Bayesian persuasion}~\citep{kamenica2011bayesian, kamenica2019bayesian}, that models how an informed principal persuades an agent to act in the principal's favor.
For an algorithmic treatment of Bayesian persuasion see \citep{dughmi2016algorithmic,dughmi2016persuasion,dughmi2017algorithmic}, and for an online learning perspective see \citep{castiglioni2020online,zu2025learning}. More generally, this framework is part of the broader area of information design in theoretical economics~\citep{bergemann2016information,taneva2019information}. Finally,~\citet{Slivkins2023survey} provides a survey on the relation between Bayesian persuasion and \BIC exploration.

%% file: sec-bic.tex
\section{Classic Bayesian Incentive Compatible Exploration}
\label{section:kremer-review}

We begin by reviewing the Bayesian incentive compatible (BIC) exploration model introduced by~\citet{kremer2014implementing}. Consider a principal (learning algorithm) that interacts with a stream of $T$ myopic selfish agents. The principal and agents have a common Bayesian prior $\theta = \theta_1 \times \dots \times \theta_k$, which defines random variables $R_a$ for the reward of each action $a \in [k]$, where $[k] = \{ 1, \ldots, k \}$. Rewards are deterministic but ex ante unknown. Namely, the reward of action~$a$, is sampled from~$\theta_a$ once, and then any pull of action $a$ yields the same reward $R_a = r_a$.

At each round $t \in [T]$, a new agent $t$ arrives. The principal delivers a private message (recommended action) $\sigma_t \in [k]$ to agent $t$.\footnote{By the revelation principle, it is without loss of generality to restrict messages to recommending actions~\citep{kremer2014implementing}.} Agent $t$ then selects an action $a_t$ and receives reward $r_{a_t} \in [0,1]$. Finally, agent $t$ leaves and never returns. The principal observes both the selected action~$a_t$ and the reward~$r_{a_t}$. The expected cumulative reward is $\EE[\theta]{\sum_{t = 1}^T R_{a_t}}$. The principal's objective is to maximize the expected cumulative reward of all the agents which is equivalent to minimizing the regret $\reg_T = \sup_{a \in [k]} \EE[\theta]{ T\cdot R_a - \sum_{t = 1}^T R_{a_t} }$. For the purpose of this exposition, we relax the principal's goal to achieving sublinear regret, i.e., $\reg_T = \littleo{T}$.\footnote{The expected cumulative reward is also known as the social welfare in prior work~\citep{kremer2014implementing,cohen2019optimal,mansour2020bayesian}.}

In order to allow the principal to explore all actions and achieve optimal expected cumulative reward asymptotically, \citet{kremer2014implementing} proposed to design an \emph{incentive compatible recommendation policy} $\pi$ for the principal, such that when the principal recommends an action to an agent, the agent follows it because the recommended action (weakly) maximizes their expected reward. Formally, a recommended action $\sigma_t$ is \emph{Bayesian incentive compatible} (\BIC) for agent~$t$ if $\EE{R_{\sigma_t} - R_a \mid \sigma_t, \pi, t, \calE_{t - 1}} \geq 0$ for any action $a \in [k]$, where $\calE_{t - 1}$ is the event that all previous agents followed their recommended actions. In particular, for the case of two actions \citet{kremer2014implementing} designed an optimal incentive compatible recommendation policy that explores both actions in finite time.

We now present an example instance in which the principal knows everything the agents know (prior $\theta$, principal policy $\pi$, their place in line $t$), and a \BIC recommendation policy can use information asymmetry to explore all actions in finite time and achieve sublinear regret.

\begin{example}
\label{example:bayesian-ext-info-full-info}
    Consider an instance with two actions ($k = 2$). Let $R_1 \sim \Unif{0,1}$ and $R_2 \sim \Bern{0.4}$.
\end{example}

We use $\Bern{p}$ to denote the Bernoulli distribution with parameter $p$, and $\Unif{I}$ to denote the uniform distribution over the interval $I$.

We adapt the scheme of \citet{kremer2014implementing} to a \BIC recommendation policy for \cref{example:bayesian-ext-info-full-info} that explores all actions in three time steps and continues recommending the optimal action to subsequent agents. The policy is deferred to \cref{section:kremer-advice-policy}.

Next, \cref{example:bayesian-ext-info-principal-knows} incorporates the rewards from \cref{example:bayesian-ext-info-full-info} and considers agents that have additional information that the principal also observes.

\begin{example}
\label{example:bayesian-ext-info-principal-knows}
    Consider an instance with two actions ($k = 2$). Let $R_1 \sim \Unif{0,1}$ and $R_2 \sim \Bern{0.4}$. In this setting, each agent observes the entire history with probability $1/2$. The principal knows exactly whether each agent knows the history or not.
\end{example}

The policy designed for \cref{example:bayesian-ext-info-full-info} is no longer \BIC for \cref{example:bayesian-ext-info-principal-knows}.\footnote{This is because agents who receive exploration-driven recommendations know that they do not maximize their expected reward if they know the history.}
Nevertheless, it can be augmented so that even in this scenario, the principal can still send \BIC recommendations and explore all actions. If agent~$t$ knows the history, the principal computes an action that maximizes agent $t$'s expected reward under this external information, recommends that action, and ignores that time step $t$ in the \BIC policy for \cref{example:bayesian-ext-info-full-info}.
If agent $t$ does not know the history, the principal proceeds according to the \BIC policy. The principal explores both actions in finite expected time, and can recommend the optimal action to subsequent agents and achieve sublinear regret.

\subsection{Breaking One-sided Information Asymmetry with External Information}
\label{section:external-information}

Generalizing the setup of BIC exploration \citep{kremer2014implementing}, we consider exploration with agents who receive private external information alongside the principal's message. The picture changes significantly if the principal does not observe the external information actually received by the agent. Consider the following example.

\begin{example}
\label{example:bayesian-ext-info}
    Consider an instance with two actions ($k = 2$). Let $R_1 \sim \Unif{0,1}$ and $R_2 \sim \Bern{0.4}$. In this setting, each agent \textbf{privately} observes the entire history of action-reward pairs produced by previous agents with probability $1/2$.
\end{example}

Here, an informed agent knows the history of action-reward pairs from previous time steps, while an uninformed agent does not. Agent 1 selects action 1 because $\EE{R_1} = 0.5 > 0.4 = \EE{R_2}$. If the realized reward $R_1 < 0.4$, action~2 is strictly better for an informed agent (since $\EE{R_2} = 0.4$). Because the principal cannot distinguish between informed and uninformed agents, the only action satisfying \BIC for all types is action~1. Consequently, if $R_1 > 0.4$, the principal never explores action~2 and thus regret is linear. Later on, we see in \cref{claim:bayesian-external-information} how we can get sublinear regret for this instance.

Therefore, external information introduces a key challenge in that \BIC may be too strict of a notion to allow for exploration. Our proposed solution is to let agents choose any action that is \emph{reasonable} for them, i.e., any action that is not worse off than another action. Moreover, we let the principal send any messages with the goal of exploring all actions, without necessarily requiring that recommended actions are \BIC. In \cref{section:behavior-policies}, we capture the behavior of agents to external information and principal messages with the notion of a behavior policy, and then in \cref{section:principal-and-agent-behavior-policies} we formalize notions of reasonable behavior policies.

Going back to the setup of \cref{example:bayesian-ext-info}, the principal can explore all actions with a general policy that sends any messages, not just actions given agents are choosing reasonable actions. We make this precise in \cref{claim:bayesian-external-information} in \cref{section:separations}.

\smallskip
\noindent\textbf{Information Asymmetry. }
\cref{example:bayesian-ext-info-full-info,example:bayesian-ext-info-principal-knows,example:bayesian-ext-info} highlight the range of information asymmetry. \cref{example:bayesian-ext-info-full-info,example:bayesian-ext-info-principal-knows} are one-sided information asymmetry, i.e., the principal knows everything the agents know (though agents in \cref{example:bayesian-ext-info-principal-knows} have more knowledge than agents in \cref{example:bayesian-ext-info-full-info}) and can achieve sublinear regret using a \BIC policy. In \cref{example:bayesian-ext-info}, however, the situation changes to two-sided information asymmetry: the principal knows part of what the agent knows, but not everything, and the agent knows something the principal does not. See illustration in \cref{fig:information-asymmetry}.

Finally, we consider a simple instance, which illustrates that allowing agents to choose any reasonable action requires going beyond the standard Bayesian view, and motivates a more general setup than the one assumed in \BIC exploration.

\begin{example}
\label{example:bayesian-full-info}
    Consider an instance with two actions ($k = 2$). Let $R_1 = 1/2$ (with probability 1), and $R_2 \sim \Unif{0,1}$.
\end{example}

Initially, both actions have an expected reward $1/2$. Both actions are reasonable for agent~1, and the principal cannot force agent~1 to select a specific action. Agent~2 now needs to reason about the choice of agent~1. However, there is no prior distribution over the choice of agent~1. This suggests that in order to capture the notion of reasonable actions, we need to move beyond the standard Bayesian setting.

%% file: sec-policies.tex
\section{Model}
\label{section:behavior-policies}

In this section, we present generalized incentivized exploration with external information, and formalize general principal and agent behavior policies.

We extend the setup of \BIC exploration from \cref{section:kremer-review} with two significant changes. Now agent $t$ receives a private external information $f_t$ from a class~$\calF_t$; this external information is not observed by the principal. Moreover, the principal and agents have a common collection of priors $\Theta$ over reality, instead of a single common Bayesian prior.

The true prior~$\theta^* \in \Theta$ defines independent random variables $R_a \sim \theta^*_a$ for the reward of each action $a \in [k]$.
Similar to \citet{kremer2014implementing} and \cref{section:kremer-review}, the reward $r_a$ of action $a \in [k]$, is sampled from $\theta^*_a$ once, and then any application of action $a$ yields the same reward. The true prior $\theta^*$ is unknown at the start of the procedure.

Let $\calH_t$ denote the set of all histories observable by the principal up to and excluding timestep $t$, where a specific history $h_t \in \calH_t$ is a sequence of action-reward pairs observed by the principal up to and excluding time $t$, namely $h_t = \left((a_1, r_{a_1}), \dots, (a_{t - 1}, r_{a_{t - 1}})\right)$.

At each time step $t$, agent~$t$ receives private external information $f_t \in \calF_t$, drawn from a fully supported distribution $\theta^*_{\external, t}$ that is defined by the true prior $\theta^*$ and may be correlated with the reward distributions $\theta^*_a$. Agent~$t$ observes only the realization $f_t$, e.g., traffic delay report on the radio, while the principal cannot observe $f_t$.

\smallskip
\noindent\textbf{Information Settings.\ }
We distinguish two information settings. In the \emph{full information} setting, agents receive no external information, i.e., $\calF_t = \{ \emptyset \}$ for all $t$. In the \emph{external information} setting, the signal $f_t$ may reveal additional information to agent $t$.

\smallskip
\noindent\textbf{Principal Policies.\ }
The principal selects a policy $\pip$, defined as a sequence of mappings $\{\pip^t\}_{t=1}^T$. Each mapping  $\pip^t : \mathcal{H}_t \to \Delta(\Sigma_t)$ maps a history $h_t$ of action-reward pairs to a distribution over $\Sigma_t$, the set of permitted messages at time $t$. We use $\Delta(\mathcal{X})$ to denote the simplex of all distributions over the abstract space~$\mathcal{X}$. At time $t$, the principal samples a message $\sigma_t \sim \pip^t(h_t)$ and sends it to agent $t$. The principal's objective is to maximize the expected cumulative reward of all the agents.
Formally, the expected cumulative reward is $\EE{\sum_{t = 1}^T R_{a_t}}$, where the expectation is taken with respect to the true prior distribution $\theta^*$ and the principal policy $\pip$.

A special type of principal policy that only sends actions as messages is an \emph{advice policy}. We call a message $\sigma_t$ that recommends an action, so $\sigma_t \in [k]$, an \emph{advice}. Correspondingly, principal policies only send advice at every step $t$, so $\Sigma_t \subseteq [k]$ are \emph{advice policies}. The recommendation policies of \citet{kremer2014implementing} are advice policies where the recommended action is guaranteed to be \BIC.
We use the term advice rather than recommendation to distinguish the communication structure from incentive properties. Our definition of advice refers to the signaled action itself, regardless of whether it is incentivized. Note that the standard notion of a recommendation policy is equivalent in our framework to an \IC advice policy paired with an \IC behavior policy.

\begin{protocol}
    \renewcommand{\algorithmicrequire}{\textbf{Parameters:}}%
    \renewcommand{\algorithmicensure}{\textbf{Unknowns:}}%
    \caption{General Incentivized Exploration Protocol}
    \label{algorithm:incentivized-exploration-protocol}
    \begin{algorithmic}[1]
        \Require $k$ actions, $T$ rounds, collection of priors $\Theta$ %
        \Ensure principal policy $\pip$, agents' behavior policies $\{\pia^t\}_{t=1}^T$ %
        \State Nature selects true prior distribution $\theta^* \in \Theta$
        \State Nature samples true reward vector $r \in [0,1]^k$ from $\theta^*$
        \State Initialize $h_1 = \emptyset$
        \For{$t = 1, \dots, T$}
            \State Agent $t$ arrives
            \State The principal selects a message $\sigma_t \sim \pip^t(h_t)$ and sends $\sigma_t$ to agent $t$
            \State Agent $t$ receives their external information $f_t$ and the principal's message $\sigma_t$
            \State Agent $t$ selects action $a_t = \pia^t(\sigma_t, f_t)$ and receives reward $r_{a_t}$
            \State The principal observes action $a_t$ and the corresponding reward $r_{a_t}$
            \State The pair $(a_t, r_t)$ is appended to the history $h_t$, so $h_{t+1} = (h_t, (a_t, r_t))$
            \State Agent $t$ departs and never returns
        \EndFor
    \end{algorithmic}
    \renewcommand{\algorithmicrequire}{\textbf{Input:}}%
    \renewcommand{\algorithmicensure}{\textbf{Output:}}%
\end{protocol}

\smallskip
\noindent\textbf{Agent Behavior Policies.\ }
Each agent $t$ knows the collection of priors $\Theta$. At their own turn, agent~$t$ has access to the following data: the principal policy $\pip$ and the timestep $t$ (their place in line). They receive their external information $f_t \in \calF_t$ and the message sent by the principal $\sigma_t$. Agent~$t$, given their knowledge, selects their action according to a behavior policy $\pia^t : \Sigma_t \times \calF_t \to [k]$ that maps messages and external information to an action. The mapping $\pia^t$ is chosen by the agent based on their place in line $t$ and their knowledge of the principal policy $\pip$, while the mapping itself determines the selected action $a_t$ only based on the principal's message $\sigma_t$ and the agent's external information $f_t$.

\smallskip
\noindent\textbf{Bayesian vs. Non-Bayesian Regimes.\ }
We distinguish between Bayesian and non-Bayesian regimes from the perspective of each agent. At time $t$, we say the regime is \emph{Bayesian} if agent $t$ can assign a well-defined probability distribution to the actions taken by previous agents given the agent's knowledge. Otherwise, the regime is \emph{non-Bayesian} for agent $t$. Importantly, an interaction can become non-Bayesian even if the environment starts with a single common prior: if earlier agents faced ties among undominated actions and used arbitrary tie-breaking, or received private external information (that agent $t$ cannot model probabilistically), then agent $t$ might have no distribution over the induced history. Separately, we call the instance Bayesian when $\abs{\Theta} = 1$, so there is a single common prior, and non-Bayesian when $\abs{\Theta} > 1$, so agents may disagree about which prior governs rewards.

\smallskip
\noindent\textbf{Exploration Protocol.\ }
An incentivized exploration protocol consists of a principal policy $\pip$ and a sequence of agent behavior policies $\{\pia^t\}_{t = 1}^T$. \cref{algorithm:incentivized-exploration-protocol} outlines the execution of a general incentivized exploration protocol.

\smallskip
\noindent\textbf{Notation.\ }
For convenience, we summarize all symbols we use in \cref{table:notation} in \cref{section:notation}.

%% file: sec-main-text.tex
\section{Notions of Reasonable Behaviors}
\label{section:principal-and-agent-behavior-policies}

In this section, we define the types of principal and agent behavior policies. We extend Bayesian incentive compatibility~\citep{kremer2014implementing} to handle external information and non-Bayesian instances, then introduce strong incentive compatibility and Pareto-optimality to characterize reasonable agent behaviors. Throughout this section, we use the term ``reasonable'' to refer to any action that is undominated given the agent's knowledge.

\subsection{Incentive Compatibility}
\label{section:weak-ic}

As we saw in \cref{section:kremer-review}, in the standard Bayesian regimes~\cite{kremer2014implementing} with full information, the revelation principle implies that we may restrict the principal's messages to advice without loss of generality. However, the presence of private external information and non-Bayesian agents complicates this picture, as advice policies are not always sufficient to induce exploration.
Nevertheless, we begin by generalizing the notion of incentive compatibility to accommodate these broader settings, as well as non-Bayesian instances.
We define when a particular advice (action) $\sigma_t$ is incentive-compatible for agent $t$, and then use this to define incentive-compatible advice policies for the principal and the corresponding behavior policy for agents.

\begin{definition}[Incentive-Compatible Advice]
\label{definition:ic-advice}
    Given an advice policy $\pip$, an advice $\sigma_t$ is \emph{incentive-compatible} (\IC) if for any external information $f_t \in \calF_t$ there is no other action $a \in  [k]$, such that for all priors $\theta \in \Theta$
    \begin{equation*}
        \EE{R_{a} - R_{\sigma_t} \mid f_t, \sigma_t; \theta, \pip, t, \calE_{t - 1}} \geq 0,
    \end{equation*}
    and for some prior $\theta \in \Theta$
    \begin{equation*}
        \EE{R_{a} - R_{\sigma_t} \mid f_t, \sigma_t; \theta, \pip, t, \calE_{t - 1}} > 0,
    \end{equation*}
    where $\calE_{t - 1}$ is the event that all previous agents $1, \dots, t - 1$ followed the \IC advice.
\end{definition}

Intuitively, an advice $\sigma_t$ is \IC if it is not strictly dominated by any other action $a \in  [k]$. That is, no alternative action $a$ is at least as good as $\sigma_{t}$ for every prior and strictly better for at least one, given the history of compliance.

In the Bayesian full information setting (with no external information and a single prior $\Theta = \{\theta\}$), the definition reduces to the classical form of \BIC~\citep{kremer2014implementing}, i.e., $\EE{R_{a} - R_{\sigma_t} \mid f_t, \sigma_t; \theta, \pip, t, \calE_{t - 1}} \leq 0$ for all actions $a \in  [k]$.

\smallskip
\noindent\textbf{Incentive-Compatible Policies.\ }
An advice policy $\pip$ for the principal is an \emph{incentive-compatible advice policy} if at every timestep $t$ and every history $h_t \in \calH_t$ the advice $\sigma_t \sim \pip^t(h_t)$ is \IC. Furthermore, we let $\Pi_{\IC}(\Theta)$ denote the set of all valid \IC advice policies for a collection~$\Theta$.

A behavior policy $\pia^t$ for agent $t$ is an \emph{incentive-compatible behavior policy} if $\pia^t(\sigma_t, f_t) = \sigma_t$ for any \IC advice $\sigma_t$. That is, under an \IC behavior policy, each agent follows the \IC advice policy, even if the \IC advice is not unique. If the advice is not \IC, it is unreasonable. As a consequence, the agent will select a different action that is reasonable (undominated).

We now present several examples that illustrate how \IC advice policies interact with \IC behavior policies. We start with a warm-up full information example and \IC advice policies in the presence of ties. Recall \cref{example:bayesian-full-info}. Initially, both actions have an expected reward $1/2$. Both actions are \IC advice to agent~1, and they will therefore follow it. To maximize expected cumulative reward, the principal recommends action 2 to agent~1, then all rewards are observed, and the principal can recommend the optimal action to the subsequent agents. This is an optimal \IC advice policy for the principal.

Finally, we look at a non-Bayesian instance with two priors, where even though there is a different action that is optimal in expectation under each prior, neither action is dominated by the other.
\begin{example}
\label{example:non-bayesian-full-info}
    Consider a non-Bayesian full information instance with two actions ($k = 2$). Under prior $\theta_1$, $R_1 \sim \Bern{0.1}$ and $R_2 \sim \Bern{0.9}$. Under prior $\theta_2$, the distributions are swapped.
\end{example}

Here, both actions can be \IC advice for agent~1, and when the agents use \IC behavior policy the principal can recommend either. If the sampled reward of this action is $1$, the principal recommends it for subsequent agents. If not, the principal recommends the other action to agent 2 since it is \IC advice, and then for subsequent agents (it is also \IC advice for them).  This is an optimal \IC advice policy for the principal.

\subsection{Strong Incentive Compatibility}
\label{section:strong-ic}

Next, we define a stronger version of incentive compatibility, in which the principal's advice is the unique undominated action for each agent~$t$. Unlike \IC advice, which permits multiple undominated actions (relying on favorable tie-breaking), strong incentive compatibility (SIC) requires that the recommended action be uniquely undominated. This restriction guarantees that there is a unique action the principal can send, and the agent will follow it. However, \SIC may be infeasible in instances where a unique undominated  action does not exist.

\begin{definition}[Strong Incentive-Compatible Advice]
\label{definition:strong-ic-advice}
    Given an advice policy $\pip$, an advice $\sigma_t$ is \emph{strong incentive-compatible} (\SIC) if for any external information $f_t \in \calF_t$ there is no other action $a \in  [k]$, such that either for all priors $\theta \in \Theta$
    \begin{equation*}
        \EE{R_{a} - R_{\sigma_t} \mid f_t, \sigma_t; \theta, \pip, t, \calE_{t - 1}} \geq 0
    \end{equation*}
    or for some prior $\theta \in \Theta$
    \begin{equation*}
        \EE{R_{a} - R_{\sigma_t} \mid f_t, \sigma_t; \theta, \pip, t, \calE_{t - 1}} > 0,
    \end{equation*}
     where $\calE_{t - 1}$ is the event that all previous agents $1, \dots, t - 1$ followed the \SIC advice.
\end{definition}

We remark that for Bayesian full information instances, SIC advice was defined in \citet{mansour2020bayesian}, and in this case, our definition coincides with it. Intuitively, an advice $\sigma_t$ is \SIC if it strictly dominates every other action $a \in  [k]$. More precisely, no alternative action $a$ is at least as good as $\sigma_t$ for every prior in $\Theta$ or has higher expected reward than $\sigma_t$ for some prior $\theta \in \Theta$, given a history of \SIC advice. Therefore, under every prior $\theta \in \Theta$, the advice $\sigma_t$ is the unique reward-maximizing action. Notice that any \SIC advice $\sigma_t$ is also an \IC advice $\sigma_t$. However, in the case of multiple reasonable (undominated)  actions, a \SIC advice $\sigma_t$ does not exist, and there can only be \IC advice.

\smallskip
\noindent\textbf{Strong Incentive-Compatible Policies.\ }
We define a \SIC advice policy for the principal analogously to \cref{section:weak-ic}. An advice policy $\pip$ is \emph{strong incentive-compatible} if at every timestep $t$ and every history $h_t \in \calH_t$ the advice $\sigma_t \sim \pip^t(h_t)$ is  \SIC advice for agent $t$. We let $\Pi_{\SIC}(\Theta)$ denote the set of all valid \SIC advice policies for a collection $\Theta$. We use the shorthand \oSIC to refer to statements that apply to both \IC and \SIC, and write $\Pi_{\oSIC}(\Theta)$, when the distinction is clear.

A behavior policy $\pia^t$ for agent $t$ is a \textit{strong incentive-compatible behavior policy} if $\pia^t(\sigma_t, f_t) = \sigma_t$ for any \SIC advice $\sigma_t$. When the advice $\sigma_t$ is not \SIC, agent $t$ does not necessarily follow the advice and may select any reasonable action.
That is, under an \SIC behavior policy, each agent follows the \SIC advice policy.

An important distinction between \IC advice and \SIC advice is how they treat ``indifference'' between actions. Consider an agent with multiple undominated actions given their knowledge. The principal's advice is \IC if it is one of the undominated actions (and will be followed if a \IC behavior policy is used). In contrast, \SIC advice does not exist.

Recall the setup of \cref{example:bayesian-full-info}. Both actions have the same expected reward for agent~1, and so both are reasonable (undominated)  selections. The principal can recommend either of them as \IC advice to agent~1, and the agent will follow the advice under an \IC behavior policy. However, there is no \SIC advice for the exact same reason, and so there is no \SIC advice policy for the principal. Therefore, \SIC might be infeasible in some settings. Similar reasoning applies to the non-Bayesian full information instance in \cref{example:non-bayesian-full-info}: there are two undominated actions at time $t = 1$, so both can be \IC advice, however, there is no \SIC advice.

Next, we show that the principal can explore marginally faster using an \IC advice policy compared to a \SIC advice policy. \cref{example:bayesian-full-info-ic-sic} illustrates that while the difference is small, the stricter \SIC constraints bound the probability of recommending unexplored actions earlier than \IC constraints.

\begin{restatable}{example}{exampleSicIc}
\label{example:bayesian-full-info-ic-sic}
    Consider a Bayesian full information instance with two actions ($k = 2$). Let $R_1 \sim \Unif{0,1}$ and $R_2 \sim \Bern{0.1}$.
\end{restatable}

Agent~1 selects action~1 since the a-priori expected reward is higher, and it is therefore the only \oSIC advice. If $R_1 < 0.1$, the principal sends action 2 to agent~2, and if $R_1 \geq 0.1$, the principal sends action 2 with probability $p$, and action~1 with probability $1 - p$.
Specifically, action 2 is \IC for $p \leq 1/81$, but it only becomes \SIC when $p < 1/81$. Full derivation can be found in \cref{section:sic-example}.

\smallskip
\noindent\textbf{Regret for (Strong) Incentive-Compatible Policies.\ }
We formalize the performance of \oSIC advice policies using a regret notion.
Because the \oSIC constraints must be satisfied simultaneously for all priors $\theta \in \Theta$, we evaluate performance in the worst-case scenario. In our setting, the reward for each action is realized upon its first selection and remains fixed for all subsequent steps. Consequently, to strictly measure the loss against the best available option, we define regret as the difference between the expected cumulative reward of the principal's policy and that of the optimal action for the realized reward instance, maximized over the worst possible prior in $\Theta$ and the worst possible sequence of \oSIC behavior policies.

Given an advice policy $\pip$, the \emph{\oSIC regret of $\pip$} with respect to a collection of priors $\Theta$ is:
\begin{equation*}
\label{definition:pseudo-regret-ic-sic}
    \reg^{\oSIC}_T(\pip, \Theta) = \sup_{\{\pia^t\}_{t = 1}^T \text{ are $\oSIC$}} \sup_{\theta \in \Theta} \EE[\{R_a \sim \theta_a\}_{a \in [k]}]{ T \cdot \sup_{a \in  [k]} R_a - \EE[\pip]{ \sum_{t = 1}^T R_{\pia^t(\sigma_t, f_t)} } }.
\end{equation*}
The outer expectation is taken with respect to the randomness of the rewards. The inner expectation is taken with respect to the randomness in the principal policy. This isolates the principal policy's performance from the particular prior. In general, we are interested in principal policies with sublinear regret, i.e. $\reg^{\oSIC}_T(\pip, \Theta) = \littleo{T}$.

Because every \SIC advice policy is also \IC advice policy, we can relate \SIC regret to \IC regret.

\begin{restatable}{observation}{SICisIC}
\label{claim:sic-is-ic}
    We have $\Pi_{\SIC} \subseteq \Pi_{\IC}$. Hence, $\inf_{\pip \in \Pi_{\IC}(\Theta)} \reg^{\IC}_T(\pip, \Theta) \leq \inf_{\pip \in \Pi_{\SIC}(\Theta)} \reg^{\SIC}_T(\pip, \Theta).$
\end{restatable}

In particular, if there is a \SIC advice policy with sublinear \SIC regret, then there is an \IC advice policy with sublinear \IC regret. For the reverse direction, in \cref{section:separations} we show instances across information settings and regimes, where a \SIC advice policy does not exist, yet some advice policy achieves sublinear \IC regret.

\subsection{Pareto-optimality}
\label{section:po}

The (strong) incentive compatibility notions from before assume the principal can induce a particular action by sending advice if it exists. However, in the external information setting, \oSIC advice policies can fail, since agents who privately learn information regarding the actions may rationally deviate from any recommendation.
Recall \cref{example:bayesian-ext-info}. In \cref{section:external-information}, we showed that advice policies cannot achieve sublinear \oSIC regret.

Motivated by this, we extend \oSIC behavior policies by permitting the agent to pick any action that is reasonable given their knowledge, rather than requiring them to follow a single advice action.
In particular, this removes technical dependence on whether a unique \SIC advice exists. We capture this by the notion of a Pareto-optimal behavior policy, which provides a robust characterization of reasonable behavior even when the principal does not observe the external information of the agents. We later show in \cref{claim:bayesian-external-information} that under such a behavior policy, the principal can guarantee sublinear regret for \cref{example:bayesian-ext-info}.

\begin{figure}[tbp]
    \centering
    \begin{tikzpicture}[
        scale=0.5,
        >=Stealth,
        action/.style={circle, draw=black, fill=gray!10, inner sep=1.8pt},
        dominated/.style={circle, draw=black, fill=gray!35, inner sep=1.8pt},
        po/.style={circle, draw=toastedorange, fill=toastedorange!12, line width=1.0pt, inner sep=1.8pt},
        ic/.style={circle, draw=richblue, fill=toastedorange!12, line width=1.0pt, inner sep=1.8pt},
        label/.style={font=\small},
        note/.style={font=\small\bfseries, text=toastedorange},
        callout/.style={font=\small, align=left},
    ]
        \def\xmax{8.0}
        \def\ymax{8.0}
        \def\xgrid{8}
        \def\ygrid{8}

        \draw[step=1, gray!25, very thin] (0,0) grid (\xmax,\ymax);

        \draw[->] (0,0) -- (\xmax+0.3,0) node[right] {$\EE{R_a \mid \theta_1}$};
        \draw[->] (0,0) -- (0,\ymax+0.3) node[above] {$\EE{R_a \mid \theta_2}$};

        \foreach \x/\xlabel in {2/0.2,4/0.4,6/0.6,8/0.8} {
            \draw (\x,0) -- (\x,-0.1);
            \node[below, font=\small] at (\x,-0.1) {\xlabel};
        }
        \foreach \y/\ylabel in {2/0.2,4/0.4,6/0.6,8/0.8} {
            \draw (0,\y) -- (-0.1,\y);
            \node[left, font=\small] at (-0.1,\y) {\ylabel};
        }

        \node[dominated] (a1) at (3.0,4.0) {};
        \node[dominated] (a2) at (4.5,3.5) {};
        \node[dominated] (a3) at (4.8,4.8) {};
        \node[po]        (a4) at (6.5,4.0) {};
        \node[po]        (a5) at (3.5,6.5) {};
        \node[po]        (a6) at (5.5,5.5) {};

        \node[label, above right=-4pt of a1] {$a_1$};
        \node[label, below right=-4pt of a2] {$a_2$};
        \node[label, below right=-4pt of a3] {$a_3$};
        \node[label, below right=-4pt of a4] {$a_4$};
        \node[label, below right=-4pt of a5] {$a_5$};
        \node[label, above right=-4pt of a6] {$a_6$};

        \def\qpad{0}

        \begin{scope}[on background layer]
            \foreach \p in {a3} {
                \path[fill=gray!60, fill opacity=0.1]
                  ($(\p.center)+(\qpad,\qpad)$) rectangle (\xgrid,\ygrid);
                \path[draw=gray!80!black, draw opacity=0.55, line width=0.5pt]
                  ($(\p.center)+(\qpad,\qpad)$) rectangle (\xgrid,\ygrid);
            }

            \foreach \p in {a4,a5,a6} {
                \path[fill=toastedorange!40, fill opacity=0.08]
                  ($(\p.center)+(\qpad,\qpad)$) rectangle (\xgrid,\ygrid);
                \path[draw=toastedorange, draw opacity=0.55, line width=0.5pt]
                  ($(\p.center)+(\qpad,\qpad)$) rectangle (\xgrid,\ygrid);
            }
        \end{scope}

        \begin{scope}[shift={(0.5,0.5)}]
            \node[po] at (0.25,1.0) {};
            \node[label, anchor=west, text=toastedorange] at (0.5,1.0) {Pareto-optimal};
            \node[dominated] at (0.25,0.25) {};
            \node[label, anchor=west] at (0.5,0.25) {Dominated action};
        \end{scope}
    \end{tikzpicture}
    \caption{
    Visualization of expected rewards for agent~1 under two priors. For simplicity, we assume the full information setting. Each action $a$ is mapped to a point $(\EE{R_{a} \mid \theta_1},\EE{R_{a} \mid \theta_2})$. An action is Pareto-optimal if there is no other action that dominates it, i.e., if there are no actions in its upper-right quadrant (weakly higher under both priors and strictly higher under at least one). Action $a_3$ is dominated because action $a_6$ is in the gray upper-right quadrant corresponding to action $a_3$, while $a_6$ is Pareto-optimal because its quadrant is empty. Notice that if action $a_6$ was not in the action space, then action $a_3$ would be \PO even though actions $a_4$ and $a_5$ improve on $a_3$ in one dimension, but are worse off in the other dimension.
    }
    \label{fig:po-rewards}
\end{figure}
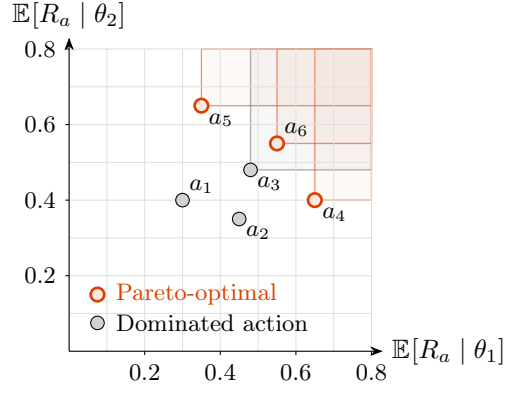

\begin{definition}[Pareto-optimal Behavior Policies]
\label{definition:po-behavioral-policies}
    Given a principal policy $\pip$, a behavior policy $\pia^t$ for agent~$t$ is \emph{Pareto-optimal} (\PO) if for any given message $\sigma_t \in \Sigma_t$ and external information $f_t \in \calF_t$ there is no action $a \in  [k]$, such that for all priors $\theta \in \Theta$ and all sequences of Pareto-optimal behavior policies $\pia^1, \dots, \pia^{t - 1}$
    \begin{equation*}
        \EE{R_{a} - R_{\pia^t(\sigma_t, f_t)} \mid f_t, \sigma_t ; \theta, \pip, \pia^1, \dots, \pia^{t - 1}, t} \geq 0,
    \end{equation*}
     and for some prior $\theta \in \Theta$ and some sequence of Pareto-optimal behavior policies $\pia^1, \dots, \pia^{t - 1}$
    \begin{equation*}
        \EE{R_{a} - R_{\pia^t(\sigma_t, f_t)} \mid f_t, \sigma_t ; \theta, \pip, \pia^1, \dots, \pia^{t - 1}, t} > 0.
    \end{equation*}
    Moreover, let $\PO(t)$ be the set of all Pareto-optimal behavior policies for agent $t$.
\end{definition}

Given what agent $t$ knows, consider the set of ``plausible worlds'' consisting of prior $\theta$, and any history consistent with the principal policy and earlier agents acting reasonably.
An action is dominated if another action weakly improves expected reward in every plausible world and strictly improves it in at least one. This is exactly how \PO behavior policies allow multiple reasonable choices, even without external information. See \cref{fig:po-rewards}
for a visual example.

We defined Pareto-optimality to capture the reasonable behaviors of the agents. However, unlike the incentive compatibility advice notions described previously, the definition of \PO imposes no constraints on the principal's policy. In fact, we allow the principal's policy to be \emph{general}: the set of messages $\Sigma_t$ the principal sends can be arbitrary and no longer restricted to actions.

\smallskip
\noindent\textbf{Obedience vs.\ autonomy.\ }
The \IC and \SIC behavior notions represent assumptions of \emph{obedience}: when the principal sends \oSIC advice, an \oSIC behavior policy follows it.
This is useful when the principal can provide a recommendation that is (strongly) incentivized, but it also relies on the strong behavioral commitment to follow the message.
In contrast, the \PO behavior notion represents \emph{autonomy}: the agent may choose any action that is not dominated given their information, and does not need to follow a recommendation even when it is undominated. Under \PO, \emph{how} the agent selects among undominated actions may depend arbitrarily on the external information $f_t$, and the principal must adjust to that.

Next, we observe that given a \SIC advice policy, the induced \SIC behavior policy is exactly the unique \PO behavior policy.

\begin{observation}
\label{observation:sic-ic-is-po}
    Given a \SIC advice policy $\pip$, for every time $t$ and every external information $f_t\in\calF_t$, the behavior policy that follows the advice, i.e., $\pia^t(\sigma_t, f_t) = \sigma_t$, is the only Pareto-optimal behavior policy for agent~$t$.
\end{observation}

In other words, under a \SIC advice policy $\pip$, for any time $t$ there is a unique reasonable action for agent~$t$ for any external information $f_t$. Therefore, any alternative action is dominated by the principal's advice in the sense of \cref{definition:po-behavioral-policies}, and so any \PO behavior policy must follow the \SIC advice.

\smallskip
\noindent\textbf{Regret for General Policies w.r.t. \PO Behavior Policies.\ }
We define regret for a principal policy with respect to agents who use PO behavior policies. Given a principal policy $\pip$, the \emph{Pareto-optimal regret} of $\pip$ with respect to the collection of priors $\Theta$ is
\begin{equation}
\label{definition:pseudo-regret-po}
    \reg^{\PO}_T(\pip, \Theta) = \sup_{\{\pia^t\}_{t = 1}^T \text{ are $\PO$}} \sup_{\theta \in \Theta} \EE[\{R_a \sim \theta_a\}_{a \in [k]}]{ T \cdot \sup_{a \in  [k]} R_a - \EE[\pip]{\sum_{t = 1}^T  R_{\pia^t(\sigma_t, f_t)} }},
\end{equation}
where the regret is taken over the worst possible sequence of \PO behavior policies and the worst possible prior $\theta \in \Theta$. Similarly to \oSIC regret, \PO regret compares the expected cumulative reward obtained from the principal's policy to the expected cumulative reward of the best fixed action policy assuming that agents use \PO behavior policies.

\begin{remark}
    \cref{section:deterministic-behavioral-policies} explains our focus on deterministic \PO behavior policies.
\end{remark}

\section{Connections between (Strong) Incentive Compatibility and Pareto-optimality}
\label{section:separations}

In this section, we compare strong incentive compatibility (\SIC), incentive compatibility (\IC), and Pareto-optimality (\PO) across different information settings (full or external information) and regimes (Bayesian or non-Bayesian).
Our aim is to understand when different behavioral notions are (1) well-defined and feasible, and (2) permit exploration with sublinear regret. Complete proofs for this section appear in \cref{section:proofs-sic-ic-po}.

\subsection{Basic Relations between Regret Notions}
\label{section:relation-po-regret}

First, we directly compare \PO regret to \SIC regret and \IC regret. Some comparisons are immediate from inclusion between behavior policy classes, while others depend on the information setting faced by the agents. Recall that in \cref{observation:sic-ic-is-po} we saw that given a \SIC advice policy $\pip$, the notion of \SIC behavior policy coincides with the notion of a \PO behavior policy. This implies the following relation between \SIC regret and \PO regret.

\begin{observation}
\label{observation:sic-po}
    For any instance, we have $\inf_{\pip} \reg^{\PO}_T(\pip, \Theta) \leq \inf_{\pip \in \Pi_{\SIC}} \reg^{\SIC}_T(\pip, \Theta).$
\end{observation}

Note that the inequality above can be strict. \cref{claim:bayesian-external-information} shows that there exists an external-information instance and a principal policy with sublinear \PO regret, but any \SIC advice policy has linear \SIC regret (hence, $\inf_{\pip} \reg^{\PO}_T(\pip, \Theta) < \inf_{\pip \in \Pi_{\SIC}} \reg^{\SIC}_T(\pip, \Theta)$).
Therefore, sublinear \SIC regret implies sublinear \PO regret. In contrast, in the full information setting, we have a reverse relation with \PO: the optimal \PO regret upper bounds the optimal \IC regret. More precisely, a principal can design an \IC advice policy by recommending an action that some \PO behavior policy would take under the same history.

\begin{restatable}{theorem}{ICPOFullInfo}
\label{claim:ic-po-full-info}
    In the full information setting we have $\inf_{\pip \in \Pi_{\IC}} \reg^{\IC}_T(\pip, \Theta) \leq \inf_{\pip} \reg^{\PO}_T(\pip, \Theta).$
\end{restatable}

Moreover, the inequality can be strict: \cref{claim:non-bayesian-full-information-no-po} establishes that there exists an instance and a \IC advice policy with sublinear \IC regret, but any principal policy has linear \PO regret (hence, $\inf_{\pip \in \Pi_{\IC}} \reg^{\IC}_T(\pip, \Theta) < \inf_{\pip} \reg^{\PO}_T(\pip, \Theta)$).

Therefore, in the full information setting, a principal policy with sublinear \PO regret implies an advice policy with sublinear \IC regret.
Later, we see that this implication can fail in the external-information setting.

\subsection{Bayesian Full Information Setting}
\label{section:bayesian-full-info}

In the classical Bayesian full information setting, we exhibit an instance in which strong incentive compatibility is infeasible, while incentive compatibility and Pareto-optimality allow for exploration with sublinear regret.

\begin{restatable}{theorem}{bayesianFullInfo}
\label{claim:bayesian-full-information}
    There exists a Bayesian full information instance, such that (1) there is no \SIC advice policy, (2) there exists an advice policy with sublinear \IC regret, and (3) there exists a general policy with sublinear \PO regret.
\end{restatable} %
\begin{proof}[Proof sketch]
    Consider the Bayesian full information instance in \cref{example:bayesian-full-info}. For (1), because neither action dominates the other for agent~1, the principal cannot send \SIC advice.
    This rules out any \SIC advice policy.

    For (3), when agents use \PO behavior policies, the above principal policy fails because the principal cannot force agent~1 to choose any specific action; both actions are reasonable for agent~1. In this case, the principal can send messages that do not necessarily contain an action. If agent~1 chooses action~2, the principal sends the observed reward. Otherwise, if agent~1 chooses action~1, the principal sends a uniformly random number in $(1/2, 1]$. This makes action~2 the only reasonable action for agent~2 by making action~1 strictly dominated given the message. Therefore, after the second agent, the principal explores all actions and can recommend the optimal action to all subsequent agents and achieve sublinear regret. Combining (3) with \cref{claim:ic-po-full-info} implies (2).
\end{proof}

\subsection{Bayesian External Information Setting}
\label{section:bayesian-external-info}

Next, we consider the effect of external information, when agents may observe additional signals unknown to the principal.
External information can make both \SIC and \IC advice infeasible.

In contrast, under \PO behavior policies (which are a weaker requirement from the agents than \IC behavior policies), general principal policies allow agents to select reasonable actions own their own, and thus always remain well-defined and can sometimes induce exploration in cases where even \IC advice policies cannot. We exhibit an instance in which a general principal policy achieves sublinear \PO regret, even though \oSIC advice policies cannot.

\begin{restatable}{theorem}{bayesianExternalInfo}
\label{claim:bayesian-external-information}
    There exists a Bayesian external information instance, such that (1) there is no advice policy with sublinear \SIC regret, (2) there is no advice policy with sublinear \IC regret, and (3) there exists a general policy with sublinear \PO regret.
\end{restatable}

\begin{proof}[Proof sketch]
    Consider the Bayesian external information instance in \cref{example:bayesian-ext-info}. Agent 1 selects action~1 because $\EE{R_1} = 0.5 > 0.4 = \EE{R_2}$. If the realized reward $R_1 > 0.4$, action~2 is strictly dominated for an informed agent (since $\EE{R_2} = 0.4$). Because the principal cannot distinguish between informed and uninformed agents, the only advice satisfying \IC for all types is action~1. Consequently, if $R_1 > 0.4$, the principal never explores action~2. Since action~2 is optimal with probability $0.4$, the policy incurs linear \IC regret. By \cref{claim:sic-is-ic}, no advice policy achieves sublinear \SIC regret.

    However, there is a policy with sublinear \PO regret. From agent 2, if $R_1 > 0.4$, the principal recommends action~2 to explore with positive probability according to the algorithm described in~\citet{kremer2014implementing}.\footnote{Notice that while \citet{kremer2014implementing} state their results for continuous reward distributions with full support, they note that their results generalize to distributions that do not have full support and are no necessarily continuous.} While informed agents observing $R_1 > 0.4$ will reject an advice that was for action 2 (playing action~1 instead), uninformed agents find action~2 undominated given their knowledge and will follow the recommendation. Since uninformed agents arrive with constant probability ($1/2$), the principal repeats the recommendation until one accepts, ensuring exploration in constant expected time. Once $R_2$ is sampled, the principal recommends the  optimal action to all subsequent agents, achieving sublinear regret.
\end{proof}

We complement the result above with an instance in which no policy achieves sublinear \SIC, \IC or \PO regret. %

\begin{restatable}{theorem}{bayesianExternalInfoNoSublinear}
\label{claim:bayesian-external-information-no-sublinear}
    There exists a Bayesian external information instance, such that (1) there is no advice policy with sublinear \SIC regret, (2) there is no advice policy with sublinear \IC regret, and (3) there is no general policy with sublinear \PO regret.
\end{restatable}

\begin{proof}
    Consider the Bayesian external information instance in \cref{example:bayesian-ext-info}, but each agent observes the history with probability $1$, instead of $1/2$. Agent~1 chooses action 1 because it has maximal expected reward. If $R_1 > 0.4$, each subsequent agent will choose action 1 regardless of the principal's message. Hence, the principal does not explore action~2. Because this event has constant probability, regret is linear under all types of behavior policies.
\end{proof}

\subsection{Non-Bayesian Full Information Setting}
\label{section:non-bayesian-full-info}

Next, we turn to non-Bayesian instances, where agents and the principal may disagree about which prior governs the rewards.
We first describe an instance that has advice policy with sublinear \IC regret and a general principal policy with sublinear \PO regret, then show an instance in which \PO behavior policies can be too permissive even under full information, so no principal policy can achieve sublinear \PO regret.

\smallskip
\noindent\textbf{Both \IC and \PO explore.\ }
The instance in \cref{example:non-bayesian-full-info} shows that merely having multiple priors does not rule out \IC advice policies. Similar to \cref{claim:bayesian-full-information}, \SIC advice policies do not exist, while there is an advice policy with sublinear \IC regret and a general policy with sublinear \PO regret.

\begin{restatable}{theorem}{nonBayesianFullInfo}
\label{claim:non-bayesian-full-information}
There exists a non-Bayesian full information instance, such that (1) there is no \SIC advice policy, (2) there exists an advice policy with sublinear \IC regret, and (3) there exists a general policy with sublinear \PO regret.
\end{restatable}

\begin{proof}[Proof sketch]
    Consider the non-Bayesian full information instance in \cref{example:non-bayesian-full-info}.
    First, at time $t = 1$, both actions are reasonable: under prior $\theta_2$ action~1 is preferred, while under prior $\theta_1$ action~2 is preferred. Hence, no \SIC advice exists at time $t = 1$, so there is no \SIC advice policy. Next, consider the following \IC advice policy: the principal recommends action~1 to agent~1. This advice is \IC because no alternative action strictly dominates it across all priors. After observing $R_1 \in \{0, 1\}$, the principal recommends action~1 to subsequent agents if $R_1=1$ (which attains the maximal possible reward), and otherwise recommends action~2. If $R_1 = 0$, in which case action~2 strictly dominates action~1 under both priors. Therefore, by time $t = 2$ the principal identifies an action with maximal reward and recommends it to subsequent agents.

    For (3), consider the following policy for the principal: send an empty message $\sigma_1 = \emptyset$ to agent~1, who may choose any reasonable action (either action in this instance). Then for $t \geq 2$ the principal sends $\sigma_t = 1$ if the principal observes either $R_1 = 1$ or $R_2 = 0$. Otherwise, the principal sends $\sigma_t = 2$.
    Therefore, the principal achieves sublinear \PO regret.
\end{proof}

\smallskip
\noindent\textbf{\IC explores, but \PO does not.\ }
Even Pareto-optimality can be too weak to guarantee exploration in non-Bayesian regimes, even with full information. In particular, for the instance described in \cref{table:non-bayesian-no-po}, we show that \PO behavior policies may be too permissive, and the principal may not explore all actions. In contrast, if agents are using \IC behavior policies, a principal can use \IC advice policy to explore all actions. Hence, the \PO regret of any principal policy will be linear, despite the existence of an advice policy with sublinear \IC regret.

\begin{table}[tbp]
\centering
\caption{Non-Bayesian full information instance with two priors and three actions. We use $\calU\{S\}$ to denote the uniform distribution over a discrete set $S$.
}
\label{table:non-bayesian-no-po}
\begin{tabular}{@{}lcc@{}}
    \toprule
    \textbf{Reward} & \textbf{Prior $\theta_1$} & \textbf{Prior $\theta_2$} \\
    \midrule
    $R_1$ & $0$ & $0.5$ \\
    $R_2$ & $0.5$ & $0$ \\
    $R_3$ & $\calU\{0.25, 0.55\}$ & $\calU\{0.25, 0.55\}$ \\
    \bottomrule
\end{tabular}
\end{table}

\begin{restatable}{theorem}{nonBayesianFullInfoNoPO}
\label{claim:non-bayesian-full-information-no-po}
There exists a non-Bayesian full information instance, such that (1) there is no \SIC advice policy, (2) there exists an advice policy with sublinear \IC regret, and (3) there is no general policy with sublinear \PO regret.
\end{restatable}
\begin{proof}[Proof sketch]
    Consider the non-Bayesian full information instance with rewards as described in \cref{table:non-bayesian-no-po}. \SIC advice is impossible because no single action is uniquely optimal for all priors at time $t = 1$.
    Next, observe that all actions are \IC advice at time $t = 1$. The principal first recommends action~3 to agent~1. If the realized reward satisfies $R_3 =0.55$, then action~3 is optimal and the principal recommends it to subsequent agents.
    Otherwise, the principal recommends action~1 to agent~2 (who is then indifferent between action~1 and action~2), determines the prior $\theta^*$ and recommends the optimal action to subsequent agents, which yields sublinear \IC regret.

    Regarding \PO, the constant behavior policy ``always play action~3'' is \PO  for every agent under any principal policy. With probability $1/2$ the reward for action~3 is $0.25$, so with constant probability the principal cannot guarantee selecting the action with highest reward. Thus, no principal policy can have sublinear \PO regret.
\end{proof}

\noindent\textbf{Relations between all regret notions.\ } To summarize the results of this section, \Cref{figure:comparison-of-regret} illustrates the hierarchical relationships between the corresponding regret notions, and \cref{table:summary-sic-ic-po-results} provides a setting-by-setting overview of our main results, indicating whether sublinear-regret policies of each type exist.

\section{Approximate Pareto-optimality}
\label{section:approximate-pareto-optimality}

In \cref{section:po} we defined Pareto-optimality to capture the notion of reasonable behavior, that is agents avoid actions that are dominated given their knowledge. We now relax this requirement to model approximately reasonable behaviors. Approximately reasonable behaviors model friction in the agent's choice---agents are not willing to switch to a better option if it does not yield a non-negligible improvement. For example, if agents are recommended a specific route, they might choose to stick with it for convenience's sake, rather than switching to a slightly faster alternative to save a single minute. See \cref{fig:eps-po-rewards}
for a visual example.

\begin{definition}[$\varepsilon$-Approximate Pareto-optimal Behavior Policies]
\label{definition:sequential-eps-po}
    Given a principal policy $\pip$, a behavior policy $\pia^t$ for agent $t$ is \emph{$\varepsilon$-approximate Pareto-optimal} (\hepsPO) for $\varepsilon \geq 0$ if for any given message $\sigma_t \in \Sigma_t$ and external information $f_t \in \calF_t$ there is no action $a \in  [k]$, such that for all priors $\theta \in \Theta$ and all sequences of $\varepsilon$-approximate Pareto-optimal behavior policies $\pia^1, \dots, \pia^{t - 1}$,
    \begin{equation*}
        \EE{R_{a} - R_{\pia^t(\sigma_t, f_t)} \mid f_t, \sigma_t ; \theta, \pip, \pia^1, \dots, \pia^{t - 1}, t} \geq \varepsilon,
    \end{equation*}
     and for some prior $\theta \in \Theta$ and some sequence of $\varepsilon$-approximate Pareto-optimal behavior policies $\pia^1, \dots, \pia^{t - 1}$
    \begin{equation*}
        \EE{R_{a} - R_{\pia^t(\sigma_t, f_t)} \mid f_t, \sigma_t ; \theta, \pip, \pia^1, \dots, \pia^{t - 1}, t} > \varepsilon.
    \end{equation*}
    Moreover, let $\hepsPO(t)$ be the set of all $\varepsilon$-approximate Pareto-optimal behavior policies for agent~$t$.
\end{definition}

Intuitively, a behavior policy $\pia^t$ is \hepsPO if the action $\pia^t(\sigma_t, f_t)$ is not \emph{$\varepsilon$-dominated}, or approximately reasonable, given that the previous agents $1, \dots, t - 1$ were using \hepsPO behavior policies.
Notice that if we set $\varepsilon = 0$, we recover \PO behavior policies (\cref{definition:po-behavioral-policies}).
Moreover, at time $t = 1$, we have $\PO(1) \subseteq \hepsPO(1)$ because there are no previous agents to quantify over. From $t \geq 2$ onward, however, the two notions can diverge because \hepsPO allows histories generated by approximately reasonable behavior policies of previous agents. The next result shows that this difference matters even in full information settings.

\begin{figure}[tbp]
    \centering
    \begin{tikzpicture}[
        scale=0.5,
        >=Stealth,
        action/.style={circle, draw=black, fill=gray!10, inner sep=1.8pt},
        dominated/.style={circle, draw=black, fill=gray!35, inner sep=1.8pt},
        po/.style={circle, draw=toastedorange, fill=toastedorange!12, line width=1.0pt, inner sep=1.8pt},
        epspo/.style={circle, draw=richblue, fill=richblue!20, line width=1.0pt, inner sep=1.8pt},
        label/.style={font=\small},
        note/.style={font=\small\bfseries, text=toastedorange},
        callout/.style={font=\small, align=left},
    ]

        \def\xmax{8.0}
        \def\ymax{8.0}
        \def\xgrid{8}
        \def\ygrid{8}

        \draw[step=1, gray!25, very thin] (0,0) grid (\xmax,\ymax);

        \draw[->] (0,0) -- (\xmax+0.3,0) node[right] {$\EE{R_a \mid \theta_1}$};
        \draw[->] (0,0) -- (0,\ymax+0.3) node[above] {$\EE{R_a \mid \theta_2}$};

        \foreach \x/\xlabel in {2/0.2,4/0.4,6/0.6,8/0.8} {
            \draw (\x,0) -- (\x,-0.1);
            \node[below, font=\small] at (\x,-0.1) {\xlabel};
        }
        \foreach \y/\ylabel in {2/0.2,4/0.4,6/0.6,8/0.8} {
            \draw (0,\y) -- (-0.1,\y);
            \node[left, font=\small] at (-0.1,\y) {\ylabel};
        }

        \node[dominated]    (a1) at (3.0,4.0) {};
        \node[dominated]    (a2) at (4.5,3.5) {};
        \node[epspo]        (a3) at (4.8,4.8) {};   %
        \node[po]           (a4) at (6.5,4.0) {};   %
        \node[po]           (a5) at (3.5,6.5) {};   %
        \node[po]           (a6) at (5.5,5.5) {};   %

        \node[label, above right=-4pt of a1] {$a_1$};
        \node[label, below right=-4pt of a2] {$a_2$};
        \node[label, below right=-4pt of a3] {$a_3$};
        \node[label, below right=-4pt of a4] {$a_4$};
        \node[label, below right=-4pt of a5] {$a_5$};
        \node[label, below right=-4pt of a6] {$a_6$};

        \def\qpad{1.0}      %

        \begin{scope}[on background layer]
            \foreach \p in {a3} {
                \path[fill=blue!40, fill opacity=0.15]
                  ($(\p.center)+(\qpad,\qpad)$) rectangle (\xgrid,\ygrid);

                \path[draw=blue!80!black, draw opacity=0.55, line width=0.5pt]
                  ($(\p.center)+(\qpad,\qpad)$) rectangle (\xgrid,\ygrid);
            }
        \end{scope}

        \begin{scope}[shift={(0.5,0.5)}]
            \node[po] at (0.25,1.75) {};
            \node[label, anchor=west, text=toastedorange] at (0.5,1.75) {Pareto-optimal};
            \node[epspo] at (0.25,1.0) {};
            \node[label, anchor=west, text=blue!80!black] at (0.5,1.0) {$\varepsilon$-Pareto-optimal};
            \node[dominated] at (0.25,0.25) {};
            \node[label, anchor=west] at (0.5,0.25) {Dominated action};
        \end{scope}
    \end{tikzpicture}
    \caption{Visualization of mean rewards for agent~1 under two priors with Pareto-optimal and $\varepsilon$-approximately Pareto-optimal actions, where $\varepsilon = 0.1$. For simplicity, we assume the full information setting and that the principal's signal $\sigma_1 = \emptyset$ is empty. Each action $a$ is mapped to the point $(\EE{R_a \mid \theta_1},\EE{R_a \mid \theta_2})$. Colored points denote \PO actions as in \cref{fig:po-rewards}. The blue point $a_3$ is not \PO, but it is $\varepsilon$-\PO because no other action is better by more than $\varepsilon$ for both priors. The blue shaded region contains any action that would be better than $a_3$ by at least $\varepsilon$ in expectation under both priors; in the picture it is empty, so $a_3$ is not $\varepsilon$-dominated.
    }
    \label{fig:eps-po-rewards}
\end{figure}
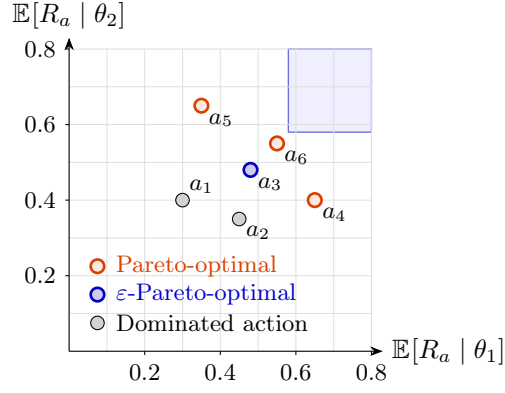

\begin{restatable}{theorem}{POepsPOIncomparable}
\label{proposition:po-to-epspo-sequential}
    Consider the full information setting and any $\varepsilon > 0$. At time $t = 1$, for any principal policy $\pip$, we have $\PO(1) \subseteq \hepsPO(1)$. However, there exists a principal policy $\pip$ and a problem instance, such that
    {$\PO(t) \not\subseteq \hepsPO(t)$ and $\hepsPO(t) \not\subseteq \PO(t)$} for some time $t > 1$.
\end{restatable}

\begin{proof}
    At time $t = 1$, there are no previous agents, so the quantification over sequences of earlier behavior policies is vacuous in both \cref{definition:po-behavioral-policies} and \cref{definition:sequential-eps-po}. Therefore the only remaining constraints are with respect to $\Theta$, and we get $\PO(1) \subseteq \hepsPO(1)$ for every $\varepsilon \geq 0$.

    Now, we turn to the case $t > 1$. Consider a single prior with rewards: $R_1 = 0.5$, $R_2 = 0.5 + \varepsilon / 3$, and $R_3 = 0.5 + 2 \varepsilon$ with probability $1 / (4 \varepsilon + 1)$ and $R_3 = 0$ with probability $4\varepsilon / (4 \varepsilon + 1)$. At step 1, there is a single \PO behavior policy that chooses action~2. However, notice that a behavior policy that selects any of the actions at time $t = 1$ is \hepsPO.
    Now, consider the principal policy $\pip$ that sends the reward observed at time $t$ as the message for time $t + 1$. In the \PO case, the only \PO policy for any timestep $t$ given that the principal is following $\pip$ is to select action~2. However, in the \hepsPO case, if at time $t = 1$ agent~1 chose action~3 and the realized reward was $0.5 + 2 \varepsilon$, then action~2 cannot be \hepsPO at timestep $t = 2$ given that the observed message is $\sigma_2 = 0.5 + 2 \varepsilon$.
\end{proof}

We define \hepsPO regret similarly to \PO regret. Given a principal policy $\pip$, the \emph{$\varepsilon$-approximate Pareto-optimal regret} of $\pip$ with respect to the collection of priors $\Theta$ is
\begin{equation*}
    \reg^{\hepsPO}_T(\pip, \Theta) = \sup_{\{\pia^t\}_{t = 1}^T \text{ are $\hepsPO$}} \sup_{\theta \in \Theta} \EE[\{R_a \sim \theta_a\}_{a \in [k]}]{ T \cdot \sup_{a \in  [k]} R_a - \EE[\pip]{\sum_{t = 1}^T  R_{\pia^t(\sigma_t, f_t)} }}.
\end{equation*}
Namely, the regret is taken over the worst possible sequence of \hepsPO behavior policies and the worst possible prior $\theta \in \Theta$.

\cref{proposition:po-to-epspo-sequential} shows there exist (full information) instances where a \PO behavior policy for agent~$t$ is not an \hepsPO behavior policy for the same agent, and vice versa.
We show that this can significantly alter the achievable optimal regret.
There are scenarios where a principal policy can achieve sublinear \PO regret, while when agents use \hepsPO behavior policies any general principal policy cannot achieve sublinear \hepsPO regret.

\begin{restatable}{theorem}{POSublinearEpsPOLinear}
\label{claim:po-sublinear-epspo-linear}
    There exists a Bayesian full information instance, such that there is a principal policy with sublinear \PO regret, however, there is no principal policy with sublinear \hepsPO regret.
\end{restatable}

\begin{proof}%
    Consider a Bayesian full information instance with two actions, such that $R_1 = \Bern{1/2}$ and $R_2 = 1/2 - \varepsilon/2$ with probability $1/2$ and $R_2 = 1/2 + \varepsilon/4$ with probability $1/2$. Action 2 has expected reward $1/2 - \varepsilon / 8$.

    In the \PO setting, agent~1 selects action~1. The principal then sends $\sigma_2 = r_1$ (the realized reward of $R_1$) to agent~2. If $r_1 = 0$, agent~2 explores action~2. Otherwise, agent~2 selects action~1 because it has the highest reward possible. The principal recommends the optimal action  to subsequent agents. This policy achieves sublinear \PO regret.

    In the \hepsPO setting, agent~1 can choose either action because their expected rewards are $\varepsilon$-close. The sequence of behavior policies that only select action~2 is \hepsPO. We show this by induction, and assume it is an \hepsPO behavior policy for times $1, \dots, t - 1$. Suppose at step $t$ selecting action~2 is not \hepsPO. By \cref{definition:sequential-eps-po} it must be dominated by more than $\varepsilon$ under every possible previous sequence of \hepsPO policies. However, regardless of the message that the principal sends, under the sequence of \hepsPO for times $1, \dots, t - 1$ this is not the case because the expected rewards differ by no more than $\varepsilon$. Therefore, the principal does not explore action~1 under some sequence of \hepsPO behavior policies. Because with probability 1/2 action~1 has reward 1, the \hepsPO regret is not sublinear.
\end{proof}

\section{Discussion and Future Work}
\label{section:conclusion}

In the current work, we extended the framework of incentivized exploration beyond the traditional Bayesian full information setting. We introduced three notions that capture different types of reasonable behavior: incentive compatibility, in which agents agree to follow any reasonable action; strong incentive compatibility in which agents only agree to follow advice only if it is the only reasonable action; and Pareto-optimality, in which agents choose any reasonable action.
We showed that while strong incentive compatibility often fails when agents are faced with tied actions, possess external information, or the instance is non-Bayesian, notions like incentive compatibility and Pareto-optimality remain viable and enable sublinear regret guarantees. In particular, Pareto-optimality remains well-defined in external information settings, allowing for efficient exploration.

Our main technical  contributions analyze the feasibility of efficient exploration under different classes of reasonable behavior. In the Bayesian regime with full information, we showed that ties between optimal actions make strong incentive compatibility impossible but incentive-compatible policies that explore efficiently do exists. In non-Bayesian instances, we proved that incentive-compatible advice policies can explore efficiently, even though any general policy has linear Pareto-optimal regret. Finally, in instances with external information where agents observe private data unknown to the principal, we constructed policies with sublinear Pareto-optimal regret despite the lack of incentive-compatible solutions.

Future research could characterize the general properties of external information structures that permit efficient exploration. It would also be natural to extend our results to long-lived agents, introducing strategic considerations where users might internalize the future value of information. Moreover, we leave in-depth investigation of approximate notions of reasonable behavior policies to future work (see \cref{section:per-step-approx-po,section:eps-ic-eps-po} for two different additional formalizations). Finally, analyzing Pareto-optimal behavior under stochastic rewards remains a key open problem.

%% file: app-notation.tex
\section{Notation and Symbols}
\label{section:notation}

\begin{table}[H]
\centering
\begin{tabular}{@{}ll@{}}
    \toprule
    \textbf{Symbol} & \textbf{Description} \\
    \midrule
    $T$ & Number of rounds (time horizon) \\
    $k$ & Number of actions (arms) \\
    $t$ & Current round, $t \in [T]$ \\
    $R_a$ & Reward random variable for action $a$ \\
    $r_a$ & Realized reward for action $a$ (unknown)
    \\
    $a_t$ & Action chosen by agent $t$ at time $t$ \\
    $\Theta$ & Collection of priors (known to all) \\
    $\theta$ & A prior in $\Theta$ \\
    $\theta^*$ & True prior (unknown) \\
    $\pip$ & Principal's policy \\
    $\pia^t$ & Behavior policy for agent $t$ \\
    $\Sigma_t$ & Set of permissible messages at time $t$\\
    $\sigma_t$ & Message sent to agent $t$ \\
    $f_t$ & External/private information available to agent $t$ \\
    $\calF$ & External information class \\
    $\calE_{t-1}$ & Event that agents $1, \dots, t-1$ followed their (S)IC advice \\
    \IC & Incentive Compatible \\
    \SIC & Strong Incentive Compatible \\
    \oSIC & Either IC or SIC, or both, depending on context \\
    \PO & Pareto Optimal \\
    \hepsPO & $\varepsilon$-approximately Pareto Optimal \\
    $\reg^{\calM}_T(\pip, \Theta)$ & Regret of policy $\pip$ w.r.t. collection of priors $\Theta$ \\ & under behavior policies of type $\calM$ \\
    \bottomrule
\end{tabular}
\caption{Notations and symbols used throughout the paper.}
\label{table:notation}
\end{table}

%% file: app-kremer-policy.tex
\section{Advice Policy for \cref{example:bayesian-ext-info-full-info}}
\label{section:kremer-advice-policy}

Below we design an advice policy for \cref{example:bayesian-ext-info-full-info} in the full information setting inspired by the algorithm presented by \citet{kremer2014implementing}. First, $\sigma_1 = 1$. Agent~1 follows the advice because this is the unique reasonable action (\SIC advice). At time $t = 2$, if $R_1 \geq 0.799$, the principal sends action~1 as advice. Otherwise, if $R_1 < 0.799$ the principal sends action~2 as advice. Notice that in either case, agent~2 follows the advice because it is the unique reasonable action (\SIC advice). In particular, we have
\begin{align*}
    \EE{R_1 - R_2 \mid \sigma_2 = 1}
    &=  \EE{R_1 - R_2 \mid R_1 \in [0.799, 1]} \\
    &=  0.8995 - 0.4 \\
    &=  0.4995 \\
    &>   0
\end{align*}
and
\begin{align*}
    \EE{R_2 - R_1 \mid \sigma_2 = 2}
    &=  \EE{R_2 - R_1 \mid R_1 \in [0, 0.799)} \\
    &=  0.4 - 0.3995 \\
    &=  0.0005 \\
    &>   0
\end{align*}
The situation for agent~3 is trickier. The principal sends the following advice
\begin{align*}
\sigma_3 =
\begin{cases}
    2 & \text{if } 0.799 \leq R_1 <1 \text{ or } (R_1 < 0.799 \text{ and } R_2 = 1) \\
    1 & \text{if } R_1 < 0.799 \text{ and } R_2 = 0
\end{cases}
\end{align*}
Observe that if the principal sends $\sigma_3 = 1$, then agent 3 follows the advice:
\begin{align*}
    \EE{R_1 - R_2 \mid \sigma_3 = 1}
    &=  0.3995 - 0 \\
    &=  0.3995 \\
    &>  0
\end{align*}
If the principal sends $\sigma_3 =2$, agent~3 follows the advice because in expectation we have:
\begin{align*}
    \EE{R_2 - R_1 \mid \sigma_3 = 2}
    &= \frac{\EE{R_2 - R_1 \mid R_1 \in [0.799, 1]} \cdot  \Pr{R_1 \in [0.799, 1]}}{\Pr{R_1 \in [0.799, 1]} + \Pr{R_2 \leq 0.799, R_2 = 1}} \\
    &\qquad\qquad + \frac{\EE{R_2 - R_1 \mid R_1 \leq 0.799, R_2 = 1} \cdot  \Pr{R_1 \leq 0.799, R_2 = 1}}{\Pr{R_1 \in [0.799, 1]} + \Pr{R_1 \leq 0.799, R_2 = 1}} \\
    &= \frac{(0.4 - 1.799 / 2) \cdot 0.201 + (1 - 0.799 / 2) \cdot (0.799 \cdot 0.4) }{0.5206} \\
    &= 0.1758 \\
    &> 0
\end{align*}
Therefore, in three steps the principal can guarantee to explore all actions and recommend the optimal action for the rest of the agents.

%% file: app-sic-example.tex
\section{Derivation of \cref{example:bayesian-full-info-ic-sic}}
\label{section:sic-example}

\exampleSicIc*

Agent~1 selects action~1 since the a-priori expected reward is higher, and it is therefore the only \oSIC advice. If $R_1 < 0.1$, the principal sends action 2 as advice to agent~2, and if $R_1 \geq 0.1$, the principal sends action 2 with probability $p$ (independent of everything else), and action~1 with probability $1 - p$.
Specifically, action 2 is \IC advice for $p \leq 1/81$, but it only becomes \SIC advice when $p < 1/81$. In particular, we have:
\begin{align*}
    \EE{R_2 - R_1 \mid \sigma_2 = 2}
    &=  \EE{R_2 - R_1 \mid R_1 \geq 0.1, \sigma_2 = 2} \Pr{R_1 \geq 0.1, \sigma_2 = 2} \\
    &\qquad\qquad + \EE{R_2 - R_1 \mid R_1 < 0.1} \Pr{R_1 < 0.1} \\
    &=  \frac{(0.1 - 0.55) \cdot p \cdot 0.9 + (0.1 - 0.05) \cdot 0.1}{0.9p + 0.1} \\
    &=  \frac{0.005 - 0.405 \cdot p}{0.9p + 0.1}
\end{align*}
The second line writes out the expectation when $\sigma_2 = 2$ as the portion that comes from $R_1 < 0.1$ and the portion that comes from $R_1 \geq 0.1$ and that the principal sent action 2 with probability p.
The last line is strictly bigger than zero when $0 \leq p < 1/81$, and it is exactly zero when $p = 1/81$.

%% file: app-deterministic-behavioral-policies.tex
\section{Why Are Behavior Policies Deterministic?}
\label{section:deterministic-behavioral-policies}

In \cref{section:behavior-policies} we defined a behavior policy $\pia^t$ for agent $t$ as a deterministic mapping from principal messages and external information to action. However, it is possible that agents would like to use mixed strategies, i.e. policies $\bar{\pi}_\mathbf{a}^t : \Sigma_t \times \calF \to \Delta([k])$. Below we show that any randomized behavior policy that is Pareto optimal, must be outputting a distribution supported on the set of deterministic Pareto optimal actions. This justifies restricting attention to deterministic behavior policies when studying Pareto-optimality in \cref{section:po}.

\begin{claim}[Randomized Pareto Optimal Behavior Policies are Supported on the Set of Deterministic Pareto Optimal Actions]
\label{claim:deterministic-po-suffices}
    Consider any principal policy $\pip$ and agent $t$ with message $\sigma_t$ and external information $f_t$. Let $\calA_{\PO}(\sigma_t, f_t) \subseteq [k]$ denote the set of deterministic actions that are Pareto optimal for agent $t$ given $(\sigma_t, f_t)$ according to \cref{definition:po-behavioral-policies}. Then a randomized \PO behavior policy $\pia^t$ outputs a distribution $p \in \Delta{[k]}$ over actions, such that $\supp(p) \subseteq \calA_{\PO}(\sigma_t, f_t)$.
\end{claim}

\begin{proof}
    Suppose there exists a randomized \PO behavior policy $\pia^t$ that assigns positive probability to some action $\hat{a} \notin \calA_{\PO}(\sigma_t, f_t)$. Therefore, by definition, action $\hat{a}$ is not \PO, so there exists some action $a^* \in [k]$ such that for all priors $\theta \in \Theta$ and all sequences of \PO behavior policies we have $\pia^1, \ldots, \pia^{t-1}$, we have $\EE{R_{a^*} - R_{\hat{a}} \mid f_t, \sigma_t; \theta, \pip, \pia^1, \ldots, \pia^{t-1}, t} \geq 0,$ and for some prior $\theta \in \Theta$ and some sequence of \PO behavior policies $\pia^1, \ldots, \pia^{t-1}$, we have $\EE{R_{a^*} - R_{\hat{a}} \mid f_t, \sigma_t; \theta, \pip, \pia^1, \ldots, \pia^{t-1}, t} > 0.$

    We define utility functions indexed by prior and sequence of \PO behavior policies: for each $\theta \in \Theta$ and each sequence of \PO behavior policies $\pia^1, \ldots, \pia^{t-1}$, let
    \begin{equation*}
        u_{\theta, \pia^1, \ldots, \pia^{t-1}}(a) = \EE{R_{a} \mid f_t, \sigma_t; \theta, \pip, \pia^1, \ldots, \pia^{t-1}, t}.
    \end{equation*}
    Next, for a randomized policy $\pia^t$ that outputs a distribution $p$ over actions, the expected utility is
    \begin{equation*}
        \bar{u}_{\theta, \pia^1, \ldots, \pia^{t-1}}(p) = \EE[a \sim p]{u_{\theta, \pia^1, \ldots, \pia^{t-1}}(a)}.
    \end{equation*}

    Because $\hat{a}$ is dominated by $a^*$, we have $u_{\theta, \pia^1, \ldots, \pia^{t-1}}(a^*) \geq u_{\theta, \pia^1, \ldots, \pia^{t-1}}(\hat{a})$ for all utility functions, with strict inequality for at least one. Consider a modified distribution $p'$ that shifts all probability mass from $\hat{a}$ to~$a^*$. Then
    \begin{equation*}
        \bar{u}_{\theta, \pia^1, \ldots, \pia^{t-1}}(p') = \bar{u}_{\theta,\pia^1, \ldots, \pia^{t-1}}(p) + p(\hat{a}) \cdot \left(u_{\theta, \pia^1, \ldots, \pia^{t-1}}(a^*) - u_{\theta, \pia^1, \ldots, \pia^{t-1}}(\hat{a})\right) \geq \bar{u}_{\theta, \pia^1, \ldots, \pia^{t-1}}(p)
    \end{equation*}
    for all utility functions, with strict inequality for at least one. Therefore, the distribution $p$ produced by $\pia^t$ is dominated by another distribution $p'$ over undominated actions. This contradicts the assumption that the randomized policy using $p$ is Pareto optimal. Therefore, any randomized \PO behavior policy must have support contained in the set of deterministic Pareto optimal actions $\calA_{\PO}(\sigma_t, f_t)$.
\end{proof}

\begin{remark}
    The result in \cref{claim:deterministic-po-suffices} shows that if the behavior policy $\pia^t$ is a randomized, its output distribution must be supported on a subset of the set of deterministic Pareto optimal actions, i.e.\ $\calA_{\PO}(\sigma_t, f_t)$. However, it is possible to design a randomized behavior policy that outputs a distribution supported on a strict subset of $\calA_{\PO}(\sigma_t, f_t)$. This would not affect its Pareto-optimality.
\end{remark}

%% file: app-proofs.tex
\section{Proofs from \cref{section:separations}}
\label{section:proofs-sic-ic-po}

\subsection{Proof from \cref{section:relation-po-regret}}

\ICPOFullInfo*

\begin{proof}
    Given an arbitrary principal policy $\pip$, we will construct an \IC advice policy $\pip^{\IC}$ whose \IC regret is bounded by the \PO regret of $\pip$. Since we are in the full-information setting, we omit $f_t = \emptyset$ from equations below. Furthermore, full-information allows the principal to compute what action an agent would pick given any message sent by the principal.

    First, we need to consider a refinement of the notion of \PO behavior policies (\cref{definition:po-behavioral-policies}) with respect to a specific sequence of past agent behaviors. A behavior policy $\pia^t$ is \PO with respect to a sequence of \PO behavior policies $\{\pia^1, \dots, \pia^{t - 1}\}$ if there is no action $a \in [k]$, such that for all priors $\theta \in \Theta$
    \begin{equation*}
        \EE{R_{a} - R_{\pia^t(\sigma_t)} \mid \sigma_t ; \theta, \pip, \pia^1, \dots, \pia^{t - 1}, t} \geq 0,
    \end{equation*}
    and for some prior $\theta \in \Theta$
    \begin{equation*}
        \EE{R_{a} - R_{\pia^t(\sigma_t)} \mid  \sigma_t ; \theta, \pip, \pia^1, \dots, \pia^{t - 1}, t} > 0.
    \end{equation*}
    We will also need the following auxiliary claim. The proof of \cref{claim:po-wrt-sequence-remains-po-wrt-all} is deferred to immediately after the proof of \cref{claim:ic-po-full-info}.

    \begin{proposition}
    \label{claim:po-wrt-sequence-remains-po-wrt-all}
        There exists a sequence of \PO behavior policies $\pia = \{\pia^t\}_{t=1}^T$, such that $\pia^t$ is also \PO with respect to the specific sequence of past policies $\{\pia^1, \dots, \pia^{t - 1}\}$.
    \end{proposition}

    Let $\pia = \{\pia^t\}_{t=1}^T$ be the sequence promised by \cref{claim:po-wrt-sequence-remains-po-wrt-all}. We define the candidate \IC advice policy $\pip^{\IC}$ as follows: at each timestep $t$, given history $h_t$, the principal samples a message $\sigma_t \sim \pip^t(h_t)$, computes the action $a_t = \pia^t(\sigma_t)$, and sends $a_t$ as the message. In particular, $\pip^{\IC}$ runs $\pip$ on the sequence of \PO behavior policies $\{\pia^1, \dots, \pia^T\}$.

    Next, we show by induction that each recommended action $a_t$ is \IC advice, and so agent $t$ who uses any \IC behavior policy follows it. For $t = 1$, agent 1 follows the principal's recommendation because it is already Pareto-optimal. Now assume this is true for agents $1, \dots, t - 1$. Recall that $\pia^t$ is a \PO behavior policy with respect to $\{\pia^1, \dots, \pia^{t - 1}\}$, then the recommended action $a_t$ is not dominated with respect to the expected rewards across all $\theta \in \Theta$ given the simulated message $\sigma_t$. As a result, a compliant agent has no strictly dominating alternative and the action $a_t$ is \IC advice. Formally, for any action $a \in [k]$ we have that for any prior $\theta \in \Theta$:
    \begin{align*}
        \EE{R_{a} - R_{\pia^t(\sigma_t)} \mid \sigma_t ; \theta, \pip^{\IC}, t, \calE_{t - 1}}
        &= \EE{R_{a} - R_{\pia^t(\sigma_t)} \mid \sigma_t ; \theta, \pip, \pia^1, \dots, \pia^{t - 1}, t},
    \end{align*}
    where $\calE_{t - 1}$ is the event that previous agents followed their \IC advice, which is guaranteed by the inductive hypothesis. Then $\pia^t(\sigma_t)$ is \IC advice because $\pia^t$ is \PO with respect to $\{\pia^1, \dots, \pia^{t - 1}\}$.
    Therefore, an agent following an \IC behavior policy will follow the recommended action.

    Finally, we can safely relate the \IC regret of $\pip^{\IC}$ to the regret of $\pip$ with respect to $\pia = \{\pia^t\}_{t = 1}^T$:
    \begin{align*}
        \reg_T(\pip^{\IC}, \Theta)
        &= \sup_{\theta \in \Theta} \EE[\{R_a \sim \theta_a\}_{a \in [k]}]{ T \cdot \sup_{a \in  [k]} R_a - \EE[\pip^{\IC}]{ \sum_{t = 1}^T R_{a_t} } } \\
        &= \sup_{\theta \in \Theta} \EE[\{R_a \sim \theta_a\}_{a \in [k]}]{ T \cdot \sup_{a \in  [k]} R_a - \EE[\pip]{ \sum_{t = 1}^T R_{\pia^t(\sigma_t)} } }.
    \end{align*}
    Observing that (1) $\pia^t \in \PO(t)$ for every timestep $t$ and (2) \PO regret takes the supremum over all \PO behavior policies, we get $\reg^{\IC}_T(\pip^{\IC}, \Theta) \leq \reg^{\PO}_T(\pip, \Theta)$.
    The upper bound holds for an arbitrary principal policy $\pip$. Taking infimum over principal policies on both sides yields the desired inequality:
    \[
        \inf_{\pip^{\IC} \in \Pi_{\IC}} \reg^{\IC}_T(\pip^{\IC}, \Theta) \leq \inf_{\pip} \reg^{\PO}_T(\pip, \Theta).
    \]
\end{proof}

\begin{proof}[Proof of \cref{claim:po-wrt-sequence-remains-po-wrt-all}]
    We inductively construct such a sequence. Let $\pia^1 \in \PO(1)$ be any \PO behavior policy for agent 1. Observe that it is trivially \PO with respect to any past sequence of \PO behavior policies; there is none of them. Now suppose the sequence $\{\pia^1, \dots, \pia^{t - 1}\}$ has the desired property. Next, we show that there exists some behavior policy $\pia^t \in \PO(t)$ that is also \PO with respect to $\{\pia^1, \dots, \pia^{t - 1}\}$. We do this by explicitly specifying $\pia^t(\sigma_t)$ for each $\sigma_t$.

    Consider the set $\calA_t$ of all actions that are \PO with respect to $\{\pia^1, \dots, \pia^{t - 1}\}$ given a message $\sigma_t$. When we lift to evaluating Pareto-optimality it with respect to any sequence of past \PO behavior policies, any action $a \in \calA_t$ can only be dominated by some other action $a' \in \calA_t$, and at least one of them must remain undominated by all others; otherwise, we get a contradiction by following any chain of dominated actions. Indeed, suppose there was some other action $a'' \not\in \calA_t$ that dominated $a$ in the sense of \cref{definition:po-behavioral-policies}, then for all priors $\theta \in \Theta$ and all \PO sequences $\hat{\pi}_{\bfa}^1, \dots, \hat{\pi}_{\bfa}^{t - 1}$
    \begin{equation*}
        \EE{R_{a''} - R_{a} \mid  \sigma_t ; \theta, \pip, \hat{\pi}_{\bfa}^1, \dots, \hat{\pi}_{\bfa}^{t - 1}, t} \geq 0,
    \end{equation*}
    and for some prior $\theta \in \Theta$ and some \PO sequence $\hat{\pi}_{\bfa}^1, \dots, \hat{\pi}_{\bfa}^{t - 1}$
    \begin{equation*}
        \EE{R_{a''} - R_{a} \mid  \sigma_t ; \theta, \pip, \hat{\pi}_{\bfa}^1, \dots, \hat{\pi}_{\bfa}^{t - 1}, t} > 0.
    \end{equation*}
    Fixating on the sequence $\{\pia^1, \dots, \pia^{t - 1}\}$, we have $\EE{R_{a''} - R_{a} \mid  \sigma_t ; \theta, \pip, \pia^1, \dots, \pia^{t - 1}, t} \geq 0$ for all priors $\theta \in \Theta$. However, $a'' \not\in \calA_t$, yet it is at least as good as $a$ under any prior, then we must have that $a \not\in \calA_t$. This is a contradiction. Therefore, there must exist some action $a \in \calA_t$ that remains \PO even when evaluated against all possible previous sequence of \PO behavior policies. We take $\pia^t(\sigma_t) = a$. Therefore, overall $\pia^t$ constitutes a valid \PO behavior policy that is also \PO with respect to the sequence of previously constructed behavior policies $\{\pia^1, \dots, \pia^{t - 1}\}$.
\end{proof}

\subsection{Proof from \cref{section:bayesian-full-info}}

\bayesianFullInfo*

\begin{proof}
    Consider the Bayesian full information instance in \cref{example:bayesian-full-info}. We derive each statement separately.

    \paragraphname{(1)\ \ Strong Incentive Compatibility}
    For the $t=1$, there is no unique dominant action since $\EE{R_1} = \EE{R_2} = 1/2$. Therefore, the principal cannot send any \SIC advice.

    \paragraphname{(2)\ \ Incentive Compatibility}
    Consider the following policy $\pi$. First, the principal sends the advice $\sigma_1 = \{2\}$ to agent~1. They follow the advice because in the absence of external information both actions have the same expected reward. Recall that $R_1 = 1/2$ with probability 1, so the reward for action 1 is fully known by the principal at the end of the first iteration. Then, for $t > 1$ the principal sends $\sigma_t = \argmax_{a \in \{1, 2\}} R_{a}$ which is \IC advice by definition.
    Because the principal learns the optimal action after $t=1$, the principal policy incurs regret $\reg_T^{\IC}(\pi,\{\theta\})\leq 1/2=\littleo{T}$.

    \paragraphname{(3)\ \ Pareto Optimality}
    Consider the following policy $\pi$. At time $t = 1$, the principal sends an empty message. Because both actions have the same expected reward, both are reasonable and agent 1 chooses either one. If agent 1 selected action 2, the principal sends $\sigma_2 = R_2$. If agent 1 selected action 1, then the principal sends $\sigma_2 \sim \mathcal{U} (1/2, 1]$. If agent 2 receives $\sigma_2 < 1/2$, it means that action 2 is strictly dominated by action 1, so they select action 1.

    Now, if $\sigma_2 > 1/2$, then agent 2 faces uncertainty over what choice agent 1 made. The principal could have sent $\sigma_2 > 1/2$, even if agent 1 selected action 1.
    However, if agent 1 selected action 2 (which is a reasonable action for them) and $\sigma_2 > 1/2$, then $R_2 > 1/2$, and so $\EE[R_2 \sim \Unif{0,1}]{R_2 \mid R_2 \geq 1/2} > 1/2=R_1$. Hence, $\pi_2^2$ is the only reasonable action of agent 2.
    By the end of iteration 2, the principal knows all the reward realizations and sends the optimal action as $\sigma_t$ for every $t \geq 3$, which is guaranteed to be the only reasonable action and incurs regret $\reg_T^{\IC}(\pi, \{\theta\}) \leq 1 = \littleo{T}$.
\end{proof}

\subsection{Proofs from \cref{section:bayesian-external-info}}

\bayesianExternalInfo*

\begin{proof}
    Consider the Bayesian external information instance in \cref{example:bayesian-ext-info}. We derive each statement separately. For this result, we first establish the Incentive Compatibility portion, and then derive the Strong Incentive Compatibility statement as a corollary.

    \paragraphname{(2)\ \ Incentive Compatibility}
    At $t=1$, the first agent has no external information. Since $\EE{R_1} = 0.5 > 0.4 = \EE{R_2}$, the unique undominated action is action~1. Thus, the only \IC advice is $\sigma_1 = 1$, and the principal must recommend it.
    For any $t > 1$, agents may receive external information $f_t$ that reveals the past history.
    Recall that an advice $\sigma_t$ is \IC only if it is undominated for \emph{all} possible realizations of external information $f_t \in \mathcal{F}_t$, as the principal does not observe $f_t$.
    Suppose the realized reward of action~1 satisfies $r_1 > 0.4$. If the principal recommends action~2, there exists an external signal $f_t$ (observing the history) under which the agent knows $r_1$. For this agent, action~1 yields a known reward $r_1 > 0.4$, while action~2 yields an expected reward of $\EE{R_2} = 0.4$. Consequently, action~2 is strictly dominated by action~1.
    Therefore, if $r_1 > 0.4$, action~2 is not \IC advice, and the principal is constrained to recommend action~1.

    Now, consider the event $E = \{r_1 > 0.4 \land r_2 = 1\}$. The probability of this event is $\Pr{R_1 > 0.4}\cdot\Pr{R_2=1} = 0.6 \cdot 0.4 = 0.24$.
    In this event, action~2 is the unique reasonable action since $r_2 = 1 > r_1$. However, as established above, the principal cannot recommend action~2 because $r_1 > 0.4$. The agents will therefore play action~1 forever.
    This results in a linear expected regret of at least $T \cdot 0.24 \cdot \EE{1 - R_1 \mid R_1 > 0.4} = \bigomega{T}$.

    \paragraphname{(1)\ \ Strong Incentive Compatibility}
    As a consequence of \cref{claim:sic-is-ic}, there is no \SIC advice policy with sublinear regret.

    \paragraphname{(3)\ \ Pareto Optimality}
    First, \cref{section:kremer-advice-policy} contains the policy for the full information version (\cref{example:bayesian-full-info}) of this instance. Next, we show how to turn this policy into a general principal policy for the external information setting of \cref{example:bayesian-ext-info}.

    \paragraphname{Policy for external information setting} Consider the following policy. The principal sends an empty message $\sigma_1 = \emptyset$ to agent~1.
    Then for $t \geq 2$, the principal follows the policy we designed above with a slight modification. If agent $t$ ignores the recommended action at time step $t$ and selects an action that is different from the principal's recommendation, then the principal ignores that time step and sends the message again  until some agent follows the advice/recommendation. An agent will potentially ignore the principal's advice if they knows the  history, that is, they are informed. Otherwise, we need to show that the policy guarantees that the recommended actions are the unique reasonable for that timestep.

    There are three possible states $S_2, S_3$ and $S_4$ for each uniformed agent to be in corresponding to steps $t = 2, 3$ and $t \geq 4$ of the full information policy. Each uninformed agent $t > 1$ has a distribution over $\{S_2, S_3, S_4\}$, representing the probability that the principal policy is in that state.
    The distribution is defined by the probability that previous agents were informed or uninformed. Due to the way we designed the full information policy each uninformed agent $t > 0$ knows that in each state the recommended action is \SIC advice. Therefore, in expectation over the states, the recommended action is \SIC advice for an uninformed agent $t > 1$ in the external information policy, and so an uninformed agent follows the advice.

    Because agents know the history with probability $1/2$, after two time steps in expectation an uninformed agent arrives. As $T \to \infty$, the principal explores both actions, and starts recommending the optimal action. Observe that informed agents who receive the  history will always select the action with the maximal reward once both actions have been explored. Hence, this principal policy attains sublinear \PO regret.
\end{proof}

\bayesianExternalInfoNoSublinear*

\begin{proof}
    Consider the Bayesian external information instance in \cref{example:bayesian-ext-info}, but each agent observes the history with probability $1$, instead of $1/2$. Agent~1 chooses action 1 because it has maximal expected reward. If $R_1 > 0.4$, each subsequent agent will choose action 1 regardless of the principal's message. Hence, the principal does not explore action~2. Because this event has constant probability, regret is linear under all types of behavior policies.
\end{proof}

\subsection{Proofs from \cref{section:non-bayesian-full-info}}

\nonBayesianFullInfo*

\begin{proof}
    Consider the non-Bayesian full information instance from \cref{example:non-bayesian-full-info}. We derive each statement separately. We repeat the reward distributions bellow.

    \begin{table}[htb]
    \centering
    \caption{Rewards for \cref{claim:non-bayesian-full-information} as presented in \cref{section:separations}.}
    \label{table:bernoulli-rewards-appendix}
    \begin{tabular}{@{}lcc@{}}
        \toprule
        \textbf{Reward R.V.} & \textbf{Prior $\theta_1$} & \textbf{Prior $\theta_2$} \\
        \midrule
        $R_1$ & $\Bern{0.1}$ & $\Bern{0.9}$ \\
        $R_2$ & $\Bern{0.9}$ & $\Bern{0.1}$ \\
        \bottomrule
    \end{tabular}
    \end{table}

    \paragraphname{(1)\ \ Strong Incentive Compatibility}
    This scenario is analogous to \cref{claim:bayesian-full-information}: for $t=1$, none of the actions is dominated by the other. In other words, for the first agent, there is no \SIC advice. Therefore, there is no \SIC advice policy.

    \paragraphname{(2)\ \ Incentive Compatibility}
    Consider the following principal policy $\pip^{\IC}$. First, the principal sends $\sigma_1 = 1$ as advice to the first agent. Agent 1 follows because $\sigma_1={1}$ is \IC advice (no action dominates it across~$\Theta$), and under an IC behavior policy agents follow any IC advice (even when IC advice is not unique). Recall that both actions follow a Bernoulli distribution for both priors. For $t > 1$ the principal sends the advice $\sigma_t = 1$ if $ R_1 = 1$, and $\sigma_t = 2$, otherwise. Because action 2 is sampled from a Bernoulli distribution, if $R_1 = 1$, then this is the best reward attainable overall. However, if $R_1 = 0$, then action 2 strictly dominates action 1 as $\EE{ R_2 - R_1 \mid \sigma_t=2 } = 1/2$. Therefore, for $t > 1$, the agent $t$ follows the advice sent by the principal. Because the principal learns the optimal actions after at most 2 iterations, the principal policy has  regret $\reg_T^{\IC}(\pi,\{\theta_1,\theta_2\})\leq 1 = \littleo{T}$.

    \paragraphname{(3)\ \ Pareto Optimality}
    Consider the following policy $\pip^{\PO}$ for the principal. Initially, send an empty message $\sigma_1 = \emptyset$ to agent 1, who may choose any action since both are reasonable. Then for $t \geq 2$, the principal sends $\sigma_t = 1$ if the principal observes either $R_1 = 1$ or $R_2 = 0$. Otherwise, the principal sends $\sigma_2 = 2$.
    Regardless of the choice made by the first agent, subsequent agents know that the recommended action is a unique reasonable action.
    Because the principal learns the optimal action after at most 2 iterations, the principal policy has  regret $\reg_T^{\PO}(\pip^{\PO},\{\theta_1,\theta_2\})\leq 1 = \littleo{T}$
\end{proof}

\nonBayesianFullInfoNoPO*

\begin{proof}
    First, we repeat the table with priors and reward distributions, then we establish each statement separately.

    \begin{table}[htb]
    \centering
    \caption{Rewards for \cref{claim:non-bayesian-full-information-no-po} as presented in \cref{section:separations}.}
    \label{table:non-bayesian-no-po-appendix}
    \begin{tabular}{@{}lcc@{}}
        \toprule
        \textbf{Reward} & \textbf{Prior $\theta_1$} & \textbf{Prior $\theta_2$} \\
        \midrule
        $R_1$ & $0$ & $0.5$ \\
        $R_2$ & $0.5$ & $0$ \\
        $R_3$ & $\calU\{0.25, 0.55\}$ & $\calU\{0.25, 0.55\}$ \\
        \bottomrule
    \end{tabular}
    \end{table}

    \paragraphname{(1)\ \ Strong Incentive Compatibility}
    At time $t = 1$, there is no single unique action that is reasonable for all priors. Therefore, there is no \SIC advice policy.

    \paragraphname{(2)\ \ Incentive Compatibility}
    Consider the following policy. The principal recommends action 3 to agent~1. Action 3 is an \IC advice because it  is not strictly dominated by some other action across all priors in $\Theta$. If $R_3 > 0.5$, then because all other actions have reward at most $0.5$, the principal has discovered the optimal arm and recommends action to all subsequent agents. If $R_3 < 0.5$, then the principal recommends action 1 to agent 2. Action 1 is an \IC advice because action 3 is now known to have reward at most 0.5, and agent 2 is indifferent between actions 1 and 2. If $R_1 = 0.5$, then the principal continues recommending action 1 until the end of the protocol. Otherwise, the principal starts recommending action 2 from time $t = 3$ until the end of the protocol. Either way, for $t\geq 3$ the optimal action is recommended and therefore the principal achieves sublinear regret.

    \paragraphname{(3)\ \ Pareto Optimality}
    We will show that the behavior policy that selects action 3 is \PO for every agent. Then with probability $1/2$, the principal is unable to learn the optimal action because under one of the priors either action 1 or action 2 has reward 0.5. Therefore, any principal policy will incur  linear \PO regret.

    For $t=1$, all three actions are \PO for agent 1 as none is dominated across all $\Theta$. Hence all three actions are undominated for agent~$1$, so in particular $3 \in \PO(1)$.
    We proceed by induction. Fix $t \ge 2$ and assume $3 \in \PO(s)$ for every $s\in\{1,\dots,t-1\}$.
    We show that $3 \in \PO(t)$, namely that action~$3$ is not dominated.

    Fix an arbitrary message $\sigma_t\in\Sigma_t$ (and recall $f_t=\emptyset$ in this instance).
    Consider the particular sequence of Pareto-optimal behavior policies for the previous agents given by
    \[
        \pia^1=3,\ \pia^2=3,\ \dots,\ \pia^{t-1}=3,
    \]
    which is valid by the inductive hypothesis.

    Suppose towards contradiction that some action $a\in\{1,2\}$ dominates action~$3$ for agent~$t$ at $(\sigma_t,\emptyset)$.
    Then, by \cref{definition:po-behavioral-policies}, for all $\theta\in\Theta$ we must have
    \[
        \EE{R_a - R_3 \mid \sigma_t;\, \theta,\pip,\pia^1,\dots,\pia^{t-1},t} \geq 0
    \]
    and strict inequality for at least one choice of $(\theta,\pia^1,\dots,\pia^{t-1})$.

    Under this sequence, before time~$t$ the principal never observes outcomes of actions~$1$ and~$2$, so conditioned on $\sigma_t$, the conditional expectation of $R_3$ is at least its minimum possible realization:
    \[
        \mathbb E\left[R_3 \mid \sigma_t;\, \theta,\pip,\pia^1,\dots,\pia^{t-1},t \right]\ \ge\ 0.25
        \qquad
        \text{for every }\theta\in\Theta.
    \]
    If $a=1$, take $\theta=\theta_1$, and then, we have $R_1=0$ surely. Hence, for $a=1$ we get
    \[
        \mathbb E\left[R_1 - R_3 \mid \sigma_t;\, \theta_1,\pip,\pia^1,\dots,\pia^{t-1},t \right]
        = - \mathbb E\left[R_3 \mid \sigma_t;\, \theta_1,\pip,\pia^1,\dots,\pia^{t-1},t \right]
        \le -0.25 < 0,
    \]
    and this contradicts the required weak-improvement inequality.
    Observe that the same holds for $a = 2$: this time we take $\theta = \theta_2$, and then, $R_2 = 0$ surely.
    Therefore, neither $a = 1$ nor $a = 2$ can dominate action~$3$ under every prior and every sequence of PO behavior policies for earlier agents.
\end{proof}

%% file: app-eps-ic.tex
\section{Approximate Incentive Compatibility and Approximate Pareto Optimality}
\label{section:eps-ic-eps-po}

\begin{figure}[tbp]
    \centering
    \begin{tikzpicture}[
        scale=0.9,
        qty/.style={rectangle, draw=black, line width=1.5pt, rounded corners, minimum height=1cm, inner sep=5pt, fill=gray!10},
        compare/.style={semithick, -},
        doublearrow/.style={thick, <->, >=Stealth},
        highlight/.style={qty, fill=richblue!10, draw=richblue},
        note/.style={text=richblue, font=\bfseries\small, align=center},
        node distance=1.2cm and 0.8cm
    ]

        \node[qty] (epsIC) {$\reg^{\varepsilon\text{-}\IC}_T(\Theta)$};
        \node[qty, right=1.25cm of epsIC] (IC) {$\reg^{\IC}_T(\Theta)$};
        \node[highlight, right=2cm of IC] (PO) {$\reg^{\PO}_T(\Theta)$};
        \node[highlight, right=1.25cm of PO] (epsPO) {$\reg^{\varepsilon\text{-}\PO}_T(\Theta)$};

        \draw[compare] (IC) -- node[midway, above] {$\leq$} (epsIC);
        \draw[compare] (epsPO) -- node[midway, above] {$\leq$} (PO);

        \draw[compare] (PO) to[bend left=35] node[midway, below, font=\small] {$\leq$ under full information} (IC);
        \draw[doublearrow] (IC) to[bend left=35] node[midway, above, font=\small] {Incomparable in general information setting} (PO);
    \end{tikzpicture}
    \caption{Comparison of regret benchmarks for approximate and exact notions of behavior policies. The $\varepsilon$-approximate notions relax their exact counterparts. In the full information setting, $\IC \leq \PO$ holds; in the general case with side information, \IC and \PO are incomparable. The regret shown is the infimum over all policies in the corresponding policy class.}
    \label{figure:comparison-of-regret-eps-ic-eps-po}
\end{figure}
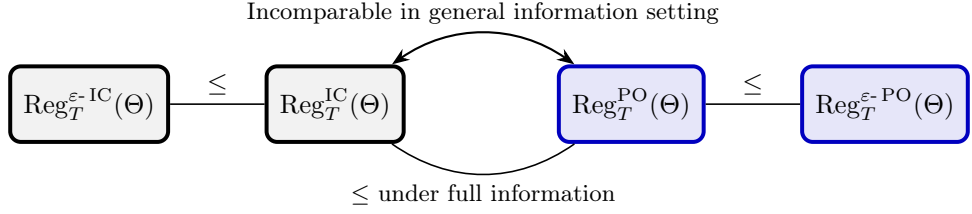

In this section we relax the notion of incentive compatible advice to capture agents who follow the recommended advice as long as it is close enough to being reasonable. In contrast to the relationship we established between \PO and \hepsPO behavior policies, when agents are required to be only $\varepsilon$-approximately incentive compatible (\epsIC), the principal explores more easily, that is any \IC advice policy is automatically an \epsIC advice policy, but not every \epsIC advice policy is an \IC advice policy.

\begin{definition}[$\varepsilon$-Approximate Incentive Compatible Advice]
\label{definition:eps-ic-advice}
    Given an advice policy $\pip$, the advice $\sigma_t$ is an \emph{$\varepsilon$-approximate incentive compatible advice} if there is no other action $a \in [k]$, such that for all priors $\theta \in \Theta$
    \begin{equation*}
        \EE{R_{a} - R_{\sigma_t} \mid f_t, \sigma_t; \theta, \pip, t, \calE_{t - 1}} \geq \varepsilon,
    \end{equation*}
    and for some prior $\theta \in \Theta$
    \begin{equation*}
        \EE{R_{a} - R_{\sigma_t} \mid f_t, \sigma_t; \theta, \pip, t, \calE_{t - 1}} > \varepsilon,
    \end{equation*}
    where $\calE_{t - 1}$ is the event that all previous agents $1, \dots, t - 1$ followed the principal's advice. In other words, an advice $\sigma_t$ is \IC advice if there is no other actions $a \in [k]$ that is at least as good as the advice $\sigma_t$ for every prior in $\Theta$ and strictly dominates the advice $\sigma_t$ by an additive $\varepsilon$ for at least one prior $\theta \in \Theta$, given that all previous agents followed the principal's advice.
\end{definition}

\begin{remark}
    While defining a notion of $\varepsilon$-approximately strong incentive compatible advice may be tempting, it has no meaningful contribution beyond the strong incentive compatible notion from \cref{section:strong-ic}. We defined \SIC advice to mitigate the effect of having multiple reasonable actions for some agent $t$. Studying the approximate version would go against this motivation. Furthermore, it is subsumed by the definition above.
\end{remark}

In the Bayesian setting with a single prior, the definition reduces to the classical form of $\varepsilon$-approximate Bayesian incentive compatibility and requires that $\EE{R_{a} - R_{\sigma_t} \mid f_t, \sigma_t; \theta, \pip, t, \calE_{t - 1}} \leq \varepsilon$ for all actions $a \in [k]$. We define $\varepsilon$-approximately \IC advice policies, $\varepsilon$-approximately \IC behavior policies and \epsIC regret analogously to their \IC counterparts. We similarly let $\Pi_{\epsIC}(\Theta)$ denote the set of all valid \epsIC advice policies for the collection~$\Theta$. Next, we establish that \epsIC advice is more permissible than \IC advice, and so \epsIC regret can improve on the \IC regret for a given problem instance.

\begin{claim}
\label{claim:eps-ic-less-than-ic}
    Consider an arbitrary prior and information setting. For any timestep $t \geq 1$ and any $\varepsilon \geq 0$, we have $\Pi_{\IC}(\Theta) \subseteq \Pi_{\epsIC}(\Theta)$. That is, any \IC advice policy is also an $\varepsilon$-approximately \IC advice policy. Therefore, $\reg^{\IC}_T(\pip, \Theta) \leq \reg^{\epsIC}_T(\pip, \Theta)$ for any principal policy $\pip$. Hence, if there exists a principal policy $\pip$ with sublinear \IC regret, then $\pip$ has sublinear \epsIC regret.
\end{claim}
\begin{proof}
    We establish each step separately.

    \paragraphname{From \IC to \epsIC}
    Consider any \IC advice policy $\pip$ and any timestep $t \geq 1$. The inequality constraints in the \IC definition for agent $t$ directly yield the inequality constraints in the \epsIC definition for $\varepsilon = 0$, and so they hold for $\varepsilon > 0$ as well. Therefore, we have that an \IC advice is also an \epsIC advice, and so we must have $\Pi_{\IC}(\Theta) \subseteq \Pi_{\epsIC}(\Theta)$.

    \paragraphname{Regret Inequality}
    Because any \IC advice policy is also an \epsIC advice policy we have: $\reg^{\IC}_T(\pip, \Theta) = \reg^{\epsIC}_T(\pip, \Theta).$
    Next, because $\Pi_{\IC}(\Theta) \subseteq \Pi_{\epsIC}(\Theta)$ and the optimal regret is taken as the infimum over those sets, we must have $\reg^{\IC}_T(\Theta) \geq \reg^{\epsIC}_T(\Theta)$.
\end{proof}

\paragraphname{Discussion on the Different Effect of Approximation}
The approximate versions of incentive compatibility and Pareto optimality interact differently with the exploration problem. For $\varepsilon$-approximate incentive compatibility (\epsIC), relaxing the inequalities makes exploration easier for the principal: we have $\Pi_{\IC}(\Theta) \subseteq \Pi_{\epsIC}(\Theta)$ (\cref{claim:eps-ic-less-than-ic}), meaning any \IC advice policy is a valid \epsIC advice policy Intuitively, agents who accept advice that is merely $\varepsilon$-close to optimal rather than exactly optimal give the principal more flexibility in which actions can be recommended.

In contrast, $\varepsilon$-approximate Pareto optimality (\hepsPO) potentially makes exploration harder. Because of the relations $\PO(t) \not\subseteq \hepsPO(t)$ and $\hepsPO(t) \not\subseteq \PO(t)$ (\cref{proposition:po-to-epspo-sequential}), the principal needs to consider a different class of possible behaviors when designing their principal policy in order to achieve low regret.
Thus, approximation plays asymmetric roles in both settings. While \epsIC relaxes constraints on what the principal can recommend, and so exploration becomes easier, \hepsPO behavior policies change the space of agent responses compared to \PO, and so exploration may be harder as we established in \cref{claim:po-sublinear-epspo-linear}.

%% file: app-per-step-pareto-optimality.tex
\section{Per-step Approximate Pareto-optimality}
\label{section:per-step-approx-po}

Here, we formalize a different notion that captures $\varepsilon$-approximately reasonable agent behaviors on the timestep scale. Under \emph{per-step} $\varepsilon$-approximate Pareto-optimality, agent~$t$ assumes that agents $1, \dots, t~-~1$ behaved exactly according to a \PO behavior policy, however, agent $t$ may select an approximately reasonable action. This notion captures reward perturbations for a particular agent, and allows the principal to be robust to small changes in the reward structure of individual agents.

\begin{definition}[Per-Step $\varepsilon$-Approximate Pareto-optimal Behavior Policies]
\label{definition:per-step-eps-po}
    Given a principal policy $\pip$ and $\varepsilon \geq 0$, a behavior policy $\pia^t$ for agent $t$ is \emph{per-step $\varepsilon$-approximate Pareto-optimal} (\epsPO) if for any given message $\sigma_t \in \Sigma_t$ and external information $f_t \in \calF_t$ there is no action $a \in  [k]$, such that for all priors $\theta \in \Theta$ and all sequences of Pareto-optimal behavior policies $\pia^1, \dots, \pia^{t - 1}$,
    \begin{equation*}
        \EE{R_{a} - R_{\pia^t(\sigma_t, f_t)} \mid f_t, \sigma_t ; \theta, \pip, \pia^1, \dots, \pia^{t - 1}, t} \geq \varepsilon,
    \end{equation*}
    and for some prior $\theta \in \Theta$ and some sequence of Pareto-optimal behavior policies $\pia^1, \dots, \pia^{t - 1}$,
    \begin{equation*}
        \EE{R_{a} - R_{\pia^t(\sigma_t, f_t)} \mid f_t, \sigma_t ; \theta, \pip, \pia^1, \dots, \pia^{t - 1}, t} > \varepsilon.
    \end{equation*}
    Let $\epsPO(t)$ be the set of all per-step $\varepsilon$-approximate Pareto-optimal behavior policies for agent $t$.
\end{definition}

Intuitively, a behavior policy $\pia^t$ is \epsPO if the action $\pia^t(\sigma_t, f_t)$ is not \emph{$\varepsilon$-dominated} given that the previous agents $1, \dots, t - 1$ were using \PO behavior policies.
Notice that if we set $\varepsilon = 0$, we recover \PO behavior policies (\cref{definition:po-behavioral-policies}). We define \epsPO regret similarly to \PO regret (\cref{definition:pseudo-regret-po}).
Given a principal policy~$\pip$, the \emph{per-step $\varepsilon$-approximate Pareto-optimal regret} of $\pip$ with respect to the collection of priors $\Theta$ is
\begin{equation*}
    \reg^{\epsPO}_T(\pip, \Theta) = \sup_{\{\pia^t\}_{t = 1}^T \text{ are $\epsPO$}} \sup_{\theta \in \Theta} \EE[\{R_a \sim \theta_a\}_{a \in [k]}]{ T \cdot \sup_{a \in  [k]} R_a - \EE[\pip]{\sum_{t = 1}^T  R_{\pia^t(\sigma_t, f_t)} }}.
\end{equation*}
Namely, the regret is taken over the worst possible sequence of \epsPO behavior policies and the worst possible prior $\theta \in \Theta$.

Next, we establish that the class of \epsPO behavior policies contains the class of \PO behavior policies for any time $t \in [T]$.

\begin{restatable}{observation}{POisAtMostEpsPO}
\label{proposition:po-is-at-most-epspo}
    Consider an arbitrary collection of priors $\Theta$ and any information setting. For any timestep $t \geq 1$ and any $\varepsilon \geq 0$, we have $\PO(t) \subseteq \epsPO(t)$. Hence, $\inf_{\pip} \reg^{\PO}_T(\pip, \Theta) \leq \inf_{\pip} \reg^{\epsPO}_T(\pip, \Theta)$.
\end{restatable}

\begin{proof}
    We establish each step separately.

    \paragraphname{From \PO to \epsPO} In both settings agent $t$ assumes that previous agents were using \PO behavior policies. Therefore, the inequality constraints in the \PO definition for agent $t$ directly yield the inequality constraints in the \hepsPO definition for $\varepsilon = 0$, and so it holds for $\varepsilon > 0$ as well. Therefore, $\PO(t) \subseteq \epsPO(t)$.

    \paragraphname{Regret Inequality}
    Let $\Pi^{\bfa}_{\PO}$ be the set of all sequences of \PO behavior policies, that is $\Pi^{\bfa}_{\PO} = \{\{\pia^t\}_{t = 1}^T \mid \pia^t \in \PO(t), t \in [T]\}$. Define $\Pi^{\bfa}_{\epsPO}$ similarly, substituting \PO with \epsPO. Now we invoke the definitions of \PO and \epsPO regret.
    \begin{align}
        \reg^{\PO}_T(\pip, \Theta)
        &=  \sup_{\{\pia^t\}_{t = 1}^T \in \Pi^{\bfa}_{\PO}} \sup_{\theta \in \Theta} \EE[R \sim \theta]{ T \cdot \sup_{a \in [k]} R_a - \EE[\pip]{\sum_{t = 1}^T  R_{\pia^t(\sigma_t, f_t)}}} \nonumber\\
        &\leq   \sup_{\{\pia^t\}_{t = 1}^T \in \Pi^{\bfa}_{\epsPO}} \sup_{\theta \in \Theta} \EE[R \sim \theta]{ T \cdot \sup_{a \in [k]} R_a - \EE[\pip]{\sum_{t = 1}^T  R_{\pia^t(\sigma_t, f_t)}}} \label{eq:po-regret-to-eps-po-regret}\\
        &=  \reg^{\epsPO}_T(\pip, \Theta) \nonumber
    \end{align}
    Line \eqref{eq:po-regret-to-eps-po-regret} holds because $\PO(t) \subseteq \epsPO(t)$ for every $t \in [T]$, so $\Pi^{\bfa}_{\PO} \subseteq \Pi^{\bfa}_{\epsPO}$.
\end{proof}

A natural question is then to characterize how much worse can \epsPO regret be compared to \PO regret. First, we observe that \epsPO regret can be higher than the \PO regret by at least $\bigtheta{\varepsilon T}$. Consider a simple Bayesian full information instance, where $R_1 = 1/2$ and $R_2 = 1/2 + \varepsilon$ with probability one. The optimal \PO regret is zero because the principal can let agents choose the optimal action (action~1) by themselves. However, the optimal \epsPO regret is $\varepsilon T$. Each agent who follows an \epsPO behavior policy may still select action~1 because it is within $\varepsilon$ of the optimal action. We formalize this observation below.

\begin{observation}
    There exists a Bayesian full information instance, such that $\inf_{\pip} \reg^{\epsPO}_T(\pip, \Theta) = \inf_{\pip} \reg^{\PO}_T(\pip, \Theta) + \bigtheta{\varepsilon T}$.
\end{observation}

Moreover, in the full information setting this relationship between \PO and \epsPO regret is tight.

\begin{restatable}{theorem}{epsPORegretComparison}
\label{proposition:po-can-be-strictly-smaller-than-epspo}
    Consider the full information setting and an instance, such that there exists a principal policy~$\pip$ that explores all actions in finite time $t^*$ when agents use \PO behavior policies. Then there exist a principal policy $\pip'$, such that $\reg^{\epsPO}_T(\pip', \Theta) = \littleo{T} + \bigo{\varepsilon T}$.
\end{restatable}

\begin{proof}
    Let $\pip$ be a policy that explores all actions in finite time $t^*$ when agents use \PO behavior policies. We construct a new policy $\pip'$. For rounds $t \leq t^*$, $\pip'$ coincides with $\pip$ on histories consistent with \PO behavior (off-path messages can be arbitrary).
    Now fix any round $t > t^*$. Under \epsPO, agent $t$ evaluates deviations assuming agents $1, \dots, t - 1$ used \PO behavior policies. Since $\pip'$ matches $\pip$ for the first $t^*$ rounds on such histories, agent $t$ believes that exploration has already happened by time $t^*$. Hence, from round $t^* + 1$ onward, $\pip'$ may send any list of action-reward pairs that is consistent with some \PO history.

    The key challenge is that the principal may not have actually observed all actions by time $t^*$, because early agents may have been \epsPO rather than \PO.
    We use rounds $t^* + 1, \dots, t^* + k$ to handle these potentially missing actions, one by one. Fix an action $a$ not yet observed by the principal at time $t^*$. There are now two cases:
    \begin{itemize}
        \item \textbf{$a$ is forceable.}
        The principal can construct a \PO-consistent list in which $a$ is more than $\varepsilon$ better than every other action.
        Then any \epsPO behavior policy must choose $a$ (otherwise $a$ would $\varepsilon$-dominate the chosen action), so $a$ is explored in that round.

        \item \textbf{$a$ is not forceable.}
        Then either $a$ is dominated under every compatible realization, or whenever $a$ is optimal its advantage over the next-best action is at most $\varepsilon$. We call these actions \emph{$\varepsilon$-benign}.
    \end{itemize}
    In either subcase, not exploring $a$ can increase per-round regret by at most $\varepsilon$.
    After at most $k$ such rounds, every action is either explored or $\varepsilon$-benign in the sense above.
    From round $t^* + k + 1$ onward, $\pip'$ sends a \PO-consistent list and induces play of a best explored action. Any remaining unexplored action can improve on this by at most $\varepsilon$, so per-round regret is at most $\varepsilon$.

    Therefore, $\pip'$ incurs only constant regret on rounds $1, \dots, t^* + k$, and at most $\varepsilon$ regret per round afterward:
    \begin{equation*}
        \reg_T^{\epsPO}(\pip', \Theta) \leq \bigo{t^* + k} + \varepsilon \cdot (T - t^* - k) = \littleo{T} + \bigo{\varepsilon T}.
    \end{equation*}
    This establishes the claim.
\end{proof}